\pdfoutput=1
\documentclass[11pt,letterpaper]{article}

\usepackage[margin=1in]{geometry}
\usepackage{lmodern}

\usepackage{mathtools}
\usepackage{amsthm,thmtools,thm-restate}
\usepackage{amsfonts,bm}
\usepackage{interval}
\intervalconfig{
  soft open fences,
  oo/.style={open},
  oc/.style={open left},
  co/.style={open right},
  cc/.style={},
}

\usepackage{enumitem}
\setlist[description]{font=\mdseries}
\setlist[enumerate,1]{label=\textup{\arabic*.},ref=\textup{\arabic*}}
\setlist[enumerate,2]{label=\textup{(\alph*)},ref=\textup{(\alph*)}}
\setlist[enumerate,3]{label=\textup{(\roman*)},ref=\textup{(\roman*)}}

\usepackage[noblocks]{authblk}
\usepackage{microtype}
\usepackage{bookmark}

\newtheorem{theorem}{Theorem}
\newtheorem{lemma}{Lemma}
\newtheorem{corollary}{Corollary}

\newtheorem*{conjecture*}{Conjecture}

\newtheorem{remark}{Remark}

\theoremstyle{definition}
\newtheorem*{definition*}{Definition}
\newtheorem{definition}{Definition}

\newcommand*{\comp}{\mathcal{C}}
\DeclareMathOperator{\poly}{poly}
\DeclareMathOperator{\sat}{sat}

\newcommand*{\indicator}[1]{\mathbf{1}\bm{\{}#1\bm{\}}}

\newcommand*{\CC}{\mathbb{C}}
\newcommand*{\hCC}{\hat{\mathbb{C}}}

\newcommand*{\NN}{\mathbb{N}}
\newcommand*{\RR}{\mathbb{R}}
\newcommand*{\TT}[1][]{\mathbb{T}_{#1}}
\newcommand*{\ZZ}{\mathbb{Z}}

\newcommand*{\cN}{\mathcal{N}}
\newcommand*{\cU}{\mathcal{U}}

\newcommand{\thx}{\theta_x}
\newcommand{\thy}{\theta_y}

\newcommand*{\csproblem}[1]{\textsc{#1}}
\newcommand*{\eps}{\varepsilon}
\newcommand*{\UG}[1][]{UG$(#1)$}
\newcommand*{\CUG}[1][]{CUG$(#1)$}

\newcommand*{\email}[1]{\href{mailto:#1}{\texttt{#1}}}

\newenvironment{itemeq}{\hspace*{\fill}$\displaystyle}{$\hspace*{\fill}}

\hypersetup{pdfauthor={Matthew Coulson, Ewan Davies, Alexandra Kolla, Viresh Patel, and Guus Regts},pdftitle={Statistical physics approaches to Unique Games}}

\author{Matthew Coulson\thanks{\email{Matthew.John.Coulson@upc.edu}}}
\affil{Department of Mathematics, Universitat Polit\`ecnica de Catalunya}

\author{Ewan Davies\thanks{\email{Ewan.Davies@colorado.edu}}}
\author{Alexandra Kolla\thanks{\email{Alexandra.Kolla@colorado.edu}}}
\affil{Department of Computer Science, University of Colorado Boulder}

\author{Viresh Patel\thanks{\email{V.S.Patel@uva.nl}}}
\author{Guus Regts\thanks{\email{G.Regts@uva.nl}}}
\affil{Korteweg-de Vries Institute for Mathematics, University of Amsterdam}

\title{Statistical physics approaches to Unique Games}

\begin{document}
\pagenumbering{roman}
\maketitle
\thispagestyle{empty}
\begin{abstract}
We show how two techniques from statistical physics can be adapted to solve a variant of the notorious Unique Games problem, potentially opening new avenues towards the Unique Games Conjecture. The variant, which we call Count Unique Games, is a promise problem in which the ``yes'' case guarantees a certain number of highly satisfiable assignments to the Unique Games instance. In the standard Unique Games problem, the ``yes'' case only guarantees at least one such assignment. We exhibit efficient algorithms for Count Unique Games based on approximating a suitable partition function for the Unique Games instance via (i) a zero-free region and polynomial interpolation, and (ii) the cluster expansion. 
We also show that a modest improvement to the parameters for which we give results with would refute the Unique Games Conjecture.
\end{abstract}

\clearpage
\pagenumbering{arabic}
\setcounter{page}{1}

\section{Introduction}

Over the last two decades, the Unique Games Problem has emerged as an obstacle to the approximability of many combinatorial optimization problems. 
More precisely, the Unique Games Conjecture (UGC) states that there is no polynomial-time algorithm to solve the Unique Games Problem within a certain performance guarantee. 
If the UGC is true, the current best-known approximation algorithms for many problems such as \csproblem{Min-2Sat-Deletion}~\cite{Khot02a}, \csproblem{Vertex Cover}~\cite{KR03}, \csproblem{Max-Cut}~\cite{KKMO04} and \csproblem{Non-Uniform Sparsest Cut}~\cite{CKKRS05,KV05} are in fact optimal. 
On the other hand, falsification of the conjecture is likely to provide powerful new algorithmic techniques that apply to many important computational problems.
The fact that either resolution of the conjecture could be an important advance in the understanding of approximation algorithms and complexity theory is one reason why the UGC has played a key role in recent theoretical computer science research.

Our main contribution is a pair of algorithms, each deeply rooted in ideas from statistical physics, that solve a natural variant of the Unique Games Problem. 
We give some important definitions before the statement of these results.

\begin{definition}
In a Unique Games problem we are given a constraint graph $G = (V,E)$, a set of colors $[k] = \{1,\dotsc,k\}$ which is referred to as the alphabet, a set of variables $\{x_u\}_{u\in V}$, one for each vertex $u$, and a set of permutations (also known as constraints) $\pi_{uv} : [k]\to [k]$, one for each edge $uv\in E$. 
We study assignments giving a color from $[k]$ to each variable $x_u$, and are interested in the number of \emph{satisfied edges} (or satisfied constraints) of the form\footnote{Formally, let $G$ be an oriented graph so each edge has a direction. Then the edge (or constraint) $u\to v$ is satisfied when $\pi_{uv}(x_u)=x_v$, or equivalently $\pi_{uv}^{-1}(x_v)=x_u$. Without orientations it is unspecified whether to use $\pi_{uv}$ or $\pi_{uv}^{-1}$ here. We suppress this complication as in other parts of the paper it is more natural to consider undirected graphs.} $\pi_{uv}(x_u)=x_v$. The \emph{value} of the assignment is the fraction of satisfied constraints. 
The \emph{value} of the Unique Games instance is the maximum fraction of constraints that can be simultaneously satisfied.

We denote by \UG[k,\eps,\delta] the promise problem consisting of a Unique Games instance with alphabet size $k$ and the promise that the instance either has value at least $1-\eps$, or has value at most $\delta$. 
To solve the problem is to correctly determine which of the two cases hold. 
\end{definition}

When the parameters are unimportant or clear from context they are omitted, and we will always be interested in $\eps,\delta\ge 0$ with $1-\eps > \delta$, otherwise the problem is ill-posed or impossible to solve. 
The Unique Games Conjecture of Khot~\cite{Khot02a} can now be stated as follows.

\begin{conjecture*}[UGC]
For any constants $\eps,\delta > 0$ with $1-\eps>\delta$, there is a positive $k(\eps,\delta)$ such that for any alphabet size $k > k(\eps, \delta)$, the problem \UG[k,\eps,\delta] is NP-hard.
\end{conjecture*}
\noindent
We note that in~\cite{DKKMS18a,DKKMS18} it was shown that \UG[1/2-\delta,\delta] is NP-hard for any $\delta>0$.

Upon translating a Unique Games problem into a form amenable to methods from statistical physics, which we elaborate upon later, a natural variant of the problem arises.

\begin{definition}
Count Unique Games, or \CUG[f,k,\eps,\delta], is the promise problem consisting of a Unique Games instance with alphabet size $k$ and the promise that the instance either has at least $(fk)^{|V|}$ colorings with value at least $1-\eps$, or that every coloring has value at most $\delta$. 
To solve the problem is to correctly determine which of the two cases hold, and when the parameters are clear from context they are omitted. 
\end{definition}

Note that with $f=k^{-1}$, \CUG[f,k,\eps,\delta] corresponds exactly to \UG[k,\eps,\delta]; and with $f=1$ the problem is easily solvable as the guarantee covers all colorings in both cases. To solve \CUG[1,k,\eps,\delta] simply requires checking the value of an arbitrary coloring.
So via the parameter $f$ CUG can be smoothly reduced in difficulty from equal to UG to trivial. We can now state our main results.

\begin{theorem}\label{thm:main,interp}
For $\eps,\delta>0$ with $1-\eps > \delta$ and any fixed integer $k\ge 3$, there exists $\Delta_0(\eps,\delta,k)$ such that for all $\Delta\ge\Delta_0$ the following holds. 

Let $G$ be an $n$-vertex, $\Delta$-regular \CUG[f,k,\eps,\delta] instance with $f \ge k^{\eps+\delta-\frac{1}{2}}$.  
Then there is a deterministic algorithm that solves CUG for $G$ in time
\[
n\exp\left(e^{O(\log^2 k)} \log^2\Delta\right)\,.
\]
\end{theorem}

\begin{theorem}\label{thm:main,cluster}
For $\eps,\delta>0$ with $1-\eps > \delta$ and any $k$ satisfying $\log k\ge \Delta^{3/2}\log\Delta$, there exists $\Delta_0(\eps,\delta)$ such that for all $\Delta\ge\Delta_0$ the following holds. 

Let $G$ be an $n$-vertex, $\Delta$-regular \CUG[f,k,\eps,\delta] instance with $f\ge k^{2\eps+2\delta-1}$.
Then there is a deterministic algorithm that solves CUG for $G$ in time $kn^{O(1)}e^{O(\Delta)}$.
\end{theorem}

To illustrate how close Theorem~\ref{thm:main,cluster} gets to an algorithm for usual Unique Games, consider an $n$-vertex, $\Delta$-regular instance of UG and suppose there exists an assignment of value $1-\eps$. 
Let $S$ be an arbitrary set of $\eps n$ vertices, and consider all assignments obtained by relabelling vertices in $S$. 
There are at most $\eps\Delta n$ constraints that could be violated by modifying the labels of vertices in $S$, hence each such assignment has value at least $1-3\eps$.
Thus there are at least $k^{\eps n}$ assignments of value $1-3\eps$, which corresponds to having parameter $f=k^{\eps-1}$ as a CUG instance. 
In summary, if we were permitted to take $f=k^{\eps/3-1}$ in Theorem~\ref{thm:main,cluster}, which is only slightly smaller than the stated $f\ge k^{2\eps+2\delta-1}$, then we would be able to refute the Unique Games Conjecture.

The idea at the heart of both Theorem~\ref{thm:main,interp} and Theorem~\ref{thm:main,cluster} is to encode a CUG problem as the problem of approximating the value of a \emph{partition function} $Z(G;w)$ which depends on the instance $G$, and is a polynomial in $w$. 
The partition function is intimately connected to statistical physics, and we use two techniques recently developed for approximating partition functions to prove our two main results. 
Theorem~\ref{thm:main,interp} is proved via \emph{polynomial interpolation}, a method due to Barvinok (see~\cite{Barvinok16} and references therein) and furthered by Patel and Regts~\cite{PR17} who improved the running time in many examples.
 Theorem~\ref{thm:main,cluster} is proved via the \emph{cluster expansion} and methods given in~\cite{BCHPT19}. 
These techniques have recently been developed and applied to the problem of approximating the partition function of the Potts model~\cite{Barvinok16,BDPR18,BCHPT19,LSS19}, random cluster model~\cite{BCHPT19}, and other models from statistical physics~\cite{Barvinok16,HPR19,PR17}.
The main conceptual advances in this paper are demonstrations that these techniques may be adapted to UG instances, cleanly handling the constraints assigned to edges that are not present in the standard Potts model. 
The main technical advance in this paper is a zero-free region for (a generalization of) the ferromagnetic Potts model partition function, that may be of independent interest. See Theorem~\ref{thm:zero-free} below.

\subsection{Paper organization}

In the following subsection we summarise related work on the UGC.
In Section~\ref{sec:solveCUG} we define our partition function and relate it to solving UG problems. 
This involves stating our algorithmic results for approximating the partition function, and proving that our main results follow from these algorithms. 
In Sections~\ref{sec:interp} and~\ref{sec:cluster} we discuss approximation algorithms for our partition function with the polynomial interpolation and cluster expansion methods respectively. These sections are entirely independent from each other.
Finally, we discuss an important open problem arising from our work and identify plausible barriers to improving our methods in Section~\ref{sec:conc}.

\subsection{Related work}

An intimate connection between the UGC and semidefinite programming (SDP) can be traced back to a seminal paper by Goemans and Williamson~\cite{GW94} on the \csproblem{Max-Cut} problem. An instance of \csproblem{Max-Cut} can be seen as a system of linear equations over $\ZZ_2$, and thus it is a Unique Games instance with alphabet size two. 
Goemans and Williamson gave an SDP-based algorithm for \csproblem{Max-Cut} which, on inputs where the maximal cut is of size $1-\eps$, produces a cut that satisfies at least a fraction $1-(2/\pi)\sqrt\eps$ of the constraints. 
A matching integrality gap was found by~\cite{Karloff96} and~\cite{FS02}, and in~\cite{KKMO04} it was proven that if the UGC is correct, then the Goemans--Williamson algorithm has the best approximation ratio that a polynomial-time algorithm for \csproblem{Max-Cut} can achieve.
Raghavendra~\cite{R08} proved that for every constraint satisfaction problem there is a polynomial time, semidefinite programming-based algorithm which, if the UGC is true, achieves the best possible approximation ratio for the problem.
These results cement the central role of the UGC in the theory of approximation algorithms.

There are also spectral algorithms that give good polynomial-time or quasi-polynomial-time approximations algorithms for large classes of Unique Games instances. 
These include expanders~\cite{AKKSTV08,MM10}, local expanders~\cite{AIMS10,RS10}, and more generally, graphs with few large eigenvalues~\cite{K10}. 
In~\cite{ABS10}, the authors gave a general sub-exponential algorithm for Unique Games based on spectral techniques.

In contrast to previous approaches to refuting the UGC, our methods use techniques from statistical physics and naturally lead to the consideration of CUG. The strengthened promise of CUG connects to an active area of research for other computational problems such as \csproblem{Sat}. 
With no assumptions finding a satisfying assignment for a 3CNF-formula is NP-hard, but how fast can we find a satisfying assignment under the assumption that many exist?
When a constant fraction of the possible assignments satisfy the formula, simply trying random assignments performs quite well; and it is an intriguing problem to match this performance with a deterministic algorithm. 
Servedio and Tan~\cite{ST17} gave such an algorithm that uses a deterministic algorithm for approximating the number of satisfying assignments of a formula as a key building block.
We note that deterministic approximate counting is at the heart of our methods too.

\section{Solving a Unique Games problem with a partition function}\label{sec:solveCUG}

In this section we define a partition function $Z(G;w)$ and describe how to solve CUG instances via knowledge of the partition function, leading to proofs of Theorems~\ref{thm:main,interp} and~\ref{thm:main,cluster}. 
We also observe how CUG naturally arises from UG in this context. 

A partition function is a mathematical object that encodes as a polynomial some weighted substructures in a graph. 
The general definition arises in the statistical physics of spin systems, and important examples include the independence polynomial and matching polynomials of a graph, see e.g.~\cite{Barvinok16}. 
Here we will only describe the partition function we define to study Unique Games instances, which is closely related to the Potts and random cluster models.

\begin{definition}\label{def:Z}
Given a Unique Games instance $G=(V, E, \pi)$, the partition function $Z(G;w)$ is a polynomial in a parameter $w\in\CC$ given as a sum of terms $w^{i}$ for each coloring of the graph with the alphabet $[k]$ that has $i$ satisfied constraints (i.e.\ that has value $i/|E|$).
That is,
\begin{equation}
Z(G; w) := \sum_{\{x_u\}_{u\in V}\in[k]^V}\;\prod_{\substack{(u,v)\in E,\\ x_v=\pi_{uv}(x_u)}}w\,,
\end{equation}
where the sum us over all assignments of colors in $[k]$ to the labels $\{x_u\}_{u\in V}$.
\end{definition}

If the permutations $\pi_{uv}$ are all the identity, then a satisfied constraint corresponds to a \emph{monochromatic edge}: both endpoints of the edge received the same color. 
In this case the above $Z(G;w)$ corresponds to the partition function of the Potts model from statistical physics. 
The relevance to Unique Games problems arises from the fact that the cases of the promise in (C)UG give contrasting bounds on $Z(G;w)$. 
When $w\ge 1$ is real, in the case that a highly satisfying assignment is guaranteed to exist we have a lower bound, and in the case that no highly satisfying assignments exist we have an upper bound. 
When the upper bound is less than the lower bound, at most one of the bounds can hold for any given instance, so knowledge of $Z(G;w)$ immediately solves the problem. 
If the bounds are sufficiently far apart an approximate value of $Z(G;w)$ suffices.  

\begin{lemma}\label{lem:solveCUG}
Consider an instance $G=(V,E)$ of \CUG[f,k,\eps,\delta], and let $\alpha>0$. 
Then to solve the instance it suffices to know any value $\xi$ satisfying $e^{-\alpha} \le Z(G;w)/\xi \le e^\alpha$ for any real $w$ such that
\[
\log w > \frac{|V|\log(1/f)+2\alpha}{(1-\eps-\delta)|E|}\,.
\]
In the case that $G$ has average degree $\Delta$ (so $2|E|=\Delta|V|$), and $\alpha=C|V|$ this becomes 
\[
\log w > \frac{2}{1-\eps-\delta}\frac{\log(1/f)+2C}{\Delta}\,.
\]
\end{lemma}

\begin{proof}
Consider any real $w\ge 1$.
Then if there are $(fk)^{|V|}$ colorings of $G$ with value $1-\eps$ we have 
\[
e^{\alpha}\xi \ge Z(G;w) \ge (fk)^{|V|}w^{(1-\eps)|E|}\,,
\]
and if every coloring has value at most $\delta$ we have 
\[
e^{-\alpha}\xi \le Z(G;w)\le k^{|V|}w^{\delta|E|}\,.
\] 
Since these bounds go in opposite directions, knowledge of $\xi$ immediately yields an solution to the problem when the implied intervals for $\xi$ are disjoint. 
This occurs precisely for $w$ as in the statement of the lemma. 
\end{proof}

We now see that increasing $f$ allows CUG instances to be solved via $Z(G;w)$ for smaller $w$, and this is how the Count variant of Unique Games naturally arises. 
We are unable to approximate $Z(G;w)$ with $w$ large enough to permit $f=1/k$ (i.e.\ refute UGC), but by slightly increasing $f$ we bring $w$ into a range amenable to our methods. 
In Section~\ref{sec:conc} we discuss the problem of how large $f$ can be while \CUG[f] is still equivalent to UG, and identify natural barriers to using algorithms for larger $w$.
To obtain Theorems~\ref{thm:main,interp} and~\ref{thm:main,cluster} we need a pair of algorithms and some calculations. 

\begin{theorem}\label{thm:alg,interp}
Let $k\in \mathbb{N}_{\geq 3}$, $\Delta\in \mathbb{N}_{\geq 3}$, and $w^*=1+(\log k -1)/\Delta$. 
Then there exists a deterministic algorithm which, given $\alpha$ satisfying $0<\alpha<\frac{2n}{e\Delta}e^{O(\log^2 k)}$,
and an $n$-vertex \UG[k] instance $G$ of maximum degree at most $\Delta$, computes a number $\xi$ satisfying $e^{-\alpha} \le Z(G;w^*)/\xi \le e^{\alpha}$ in time bounded by
\[
n\exp\left(e^{O(\log^2 k)} \log\left(\frac{n\Delta}{\alpha}e^{O(\log^2 k)}\right)\log(\Delta\sqrt k)\right)\,.
\]
\end{theorem}

\begin{theorem}\label{thm:alg,cluster}
Let $\Delta\in \mathbb{N}$ and let $\zeta=8/\sqrt{\Delta}$.
For $k\geq\exp{((18\Delta+4\Delta\log\Delta)/\zeta)}$ and $w^*=\exp((2-\zeta) \log(k)/\Delta)$ there exists an deterministic algorithm, which given $\alpha>0$ and an $n$-vertex \UG[k] instance $G$ of maximum degree at most $\Delta$, computes $\xi$ satisfying $e^{-\alpha}\le |Z(G;w^*)/\xi| \le e^\alpha$ in time bounded by $kn^{O(1)}(n/\alpha)^{O(\Delta)}$.
\end{theorem}

These results are proved in Sections~\ref{sec:interp} and~\ref{sec:cluster} respectively. 
In both cases we define a series for $\log Z(G;w)$, show that it converges, and obtain an additive approximation to it by evaluating a truncation of the series. 
Here we give the calculations that show what CUG problems we can solve with these algorithms. 

\begin{proof}[Proof of Theorem~\ref{thm:main,interp}]
Take $w^*$ as in Theorem~\ref{thm:alg,interp}.
As $\Delta\to\infty$, and with approximation error $\alpha=Cn$ for any $C<2e^{O(\log^2k)-1}/\Delta$, by Lemma~\ref{lem:solveCUG} we require 
\[
\log w^* = (1-o(1))\frac{\log k}{\Delta} > \frac{2}{1-\eps-\delta}\frac{\log(1/f)+2C}{\Delta}\,.
\]
For large enough $\Delta$, $C=\log(k)/\Delta$ is valid in Theorem~\ref{thm:alg,interp} and implies the above for any
\[
f \ge k^{\eps+\delta-\frac{1}{2}}\,.\qedhere
\]
\end{proof}
We remark that a very similar calculation gives a result for $k$ growing with $\Delta$, but we present the special case of constant $k$ here as it is instructive of our methods and permits a concise expression for $f$ and the running time.

\begin{proof}[Proof of Theorem~\ref{thm:main,cluster}]
Take $w^*$ as in Theorem~\ref{thm:alg,cluster}. 
With $\Delta\ge e^{9/2}$, $\zeta=8/\sqrt\Delta$, $k\ge \Delta^{\Delta^{3/2}}$, and $\alpha=Cn$ for some $C>0$ we choose later, by Lemma~\ref{lem:solveCUG} we require 
\[
\log w = (2-\zeta)\frac{\log k}{\Delta} > \frac{2}{1-\eps-\delta}\frac{\log(1/f)+2C}{\Delta}\,,
\]
which holds when
\[
f > e^{2C}k^{-\frac{1}{2}(2-\zeta)(1-\eps-\delta)}\,.
\]
With $\Delta\ge\Delta_0(\eps,\delta)$ and $C \le \frac{1}{4}(\eps+\delta)\log k$, this holds for $f \ge k^{2\eps+2\delta-1}$.
For large enough $\Delta_0$ (which makes $k$ sufficiently large) we can take e.g. $C=1/2$ and obtain a running time of $kn^{O(1)}e^{O(\Delta)}$.
\end{proof}

\section{Polynomial interpolation}\label{sec:interp}

The proof of Theorem~\ref{thm:alg,interp} proceeds by an influential method known as polynomial interpolation introduced by Barvinok, see e.g.~\cite{Barvinok16}. 
In our application of this method, for some real $w^*>0$ we show that there is a region $\cU\subset\CC$ containing the interval $[1,w^*]$ on which $Z(G;w)\ne0$ for any UG instance $G$ of maximum degree $\Delta$. 
The region $\cU$ is a \emph{zero-free region}, and it guarantees that a Taylor series for (a suitable modification of) $\log Z$ converges inside $\cU$.
We then approximate $Z$ by computing the coefficients of a truncation of the Taylor series.
We require that $\cU$ is independent of the size of the graph $G$ to obtain a good approximation. 
The analysis yielding the approximation from $\cU$ is rather standard, e.g.~\cite{Barvinok16,Barvinok18,PR17}, though we include it in Appendix~\ref{app:interp,details} for completeness. The main technical work is in the following theorem establishing the region $\cU$. 
We write $\cN(S, \eta)$ for an open set in $\CC$ containing the open ball of radius $\eta$ around every point in $S$.

\begin{theorem}\label{thm:zero-free}
Let $k\in\NN_{\geq 3}$ and $\Delta\in \NN_{\geq 3}$.
Then with $w^*= 1+(\log k -1)/\Delta$ there exists $\eta=\omega(\frac{1}{\Delta\log k})$ such that for any $w\in \cN\left(\left[1,w^*\right],\eta\right)$ and any \UG[k] instance $G$ of maximum degree at most $\Delta$, $Z(G;w)\neq 0$.
\end{theorem}

Our proof of Theorem~\ref{thm:zero-free} is inductive in the style of~\cite{BDPR18} which gives a zero-free region for the antiferromagnetic Potts model, though here we have a generalization of the \emph{ferromagnetic} Potts model rather than than the \emph{antiferromagnetic} Potts model that was studied in~\cite{BDPR18}. 

Consider an $n$-vertex \UG[k] instance with $G=(V,E,\pi)$ with maximum degree $\Delta$.
In order to prove our results, we will need to work more generally with the partition function with boundary conditions.
For $m>0$ and a list $W=w_1\ldots w_m$ of distinct vertices of $V$ and a list $L=\ell_1\ldots\ell_m$ of pre-assigned colours in $[k]$ for the vertices in $W$ the \emph{restricted partition function} $Z^{W}_{L}(G;w)$ is defined by
\[
Z^{W}_{L}(G;w):=\sum_{\substack{\{x_u\}_{u\in V}\in[k]^V\\\{x_u\}_{u\in V} \text{ respects }(W,L)}}\;\prod_{\substack{(u,v)\in E,\\ x_v=\pi_{uv}(x_u)}}w\,,
\]
where we say that a color assignment $\{x_u\}_{u\in V}$ respects $(W,L)$ if for all $i=1\ldots, m$ we have $x_{w_i}=\ell_i$.
As it does not vary in the steps of the proof, we will omit the parameter $w$ and write $Z^W_L(G)$ for $Z^W_L(G;w)$.
We call the vertices $w_1,\ldots,w_m$ \emph{fixed} and refer to the remaining vertices in $V$ as \emph{free} vertices. The length of $W$ (resp.\ $L$), written $|W|$ (resp.\ $|L|$) is the length of the list.
Given a list of distinct vertices $W' = w_1\ldots w_m$, and a vertex $u$ (distinct from $w_1, \ldots, w_m$) we write $W = W'u$ for the concatenated list $W = w_1\ldots w_m u$ and we use similar notation $L'\ell$ for concatenation of lists of colours.  We write $\deg(v)$ for the degree of a vertex $v$ and we write $G\setminus uv$ ($G-u$) for the graph obtained from $G$ by removing the edge $uv$ (by removing the vertex $u$). 


To prove Theorem~\ref{thm:zero-free} we consider the same statement for restricted partition functions and induct over the number of vertices whose color is not fixed by the boundary conditions.
With a strengthened induction hypothesis we can argue that unfixing the specified color of a single vertex cannot affect the value of the partition function too much and continue the induction. 
The main technical difficulties are to bound the change in angle and radius unfixing a vertex can induce in the value of the partition function (as a complex number). 

\begin{restatable}{lemma}{zerofreeinduction}
\label{lem:kcol}
Let $\Delta\in \mathbb{N}_{\geq 3}$ and let $k\in \mathbb{N}_{\geq 3}$.  
Let $c=\log k-1$ and $\alpha=\log k^{1/2}-1$.
Then there exists constants $0<\eps<\theta<\frac{\pi}{3\Delta}$ with $\eps,\theta=\omega(1/\Delta)$ and $\eta=\omega(1/(\Delta\log k))$ such that for any $w\in \mathcal{N}([1,1+c/\Delta],\eta)$ and any \UG[k] instance $G$ of maximum degree at most $\Delta$
the following hold.
\begin{enumerate}
 \item\label{itm:kcol,nonzero}
 For all lists $W$ of distinct vertices of $G$ and all lists of pre-assigned colours $L$ of length $|W|$, $Z_L^W(G) \ne 0$.
 \item\label{itm:kcol,leaf}
 For all lists $W=W'u$ of distinct vertices of $G$ such that $u$ is a leaf and any two lists $L'l$, $L'l'$ of length $W$, the following hold.
 \begin{enumerate}
 \item\label{itm:kcol,leaf,free}
 If the unique neighbour $v$ of $u$ is free,
 \begin{enumerate}
 \item\label{itm:kcol,leaf,free,angle}
  the angle between vectors $Z_{L'l}^{W'u}(G)$ and $Z_{L'l'}^{W'u}(G)$ is at most $\theta$, and
 \item\label{itm:kcol,leaf,free,length}
       \begin{itemeq}
          \frac{Z_{L'l}^{W'u}(G)}{Z_{L'l'}^{W'u}(G)}   \leq 1 + \frac{\alpha}{\Delta}\,.
       \end{itemeq}
 \end{enumerate}
 \item\label{itm:kcol,leaf,fixed}
 If the unique neighbour $v$ of $u$ is fixed,
 \begin{enumerate}
 \item\label{itm:kcol,leaf,fixed,angle}
 the angle between vectors $Z_{L'l}^{W'u}(G)$ and $Z_{L'l'}^{W'u}(G)$ is at most $\eps$, and
 \item\label{itm:kcol,leaf,fixed,length} 
       \begin{itemeq}
          \frac{|Z_{L'l}^{W'u}(G)|}{|Z_{L'l'}^{W'u}(G)|}   \leq 1 + \frac{c}{\Delta}\,.
       \end{itemeq}
 \end{enumerate}
 \end{enumerate}
 \item\label{itm:kcol,nonleaf}
 For all lists $W = W'u$ of distinct vertices of $G$, and for all lists of pre-assigned colours $L'$ of length $|W'|$, let $d$ be the number of free neighbours of $u$, and let $b=\Delta-d$. Then for any pair of colours $l,l'$,
 \begin{enumerate}
 \item\label{itm:kcol,nonleaf,angle} 
 the angle between vectors $Z_{L'l}^{W'u}(G)$ and $Z_{L'l'}^{W'u}(G)$ is at most $d \theta + b\eps$, and 
 \item\label{itm:kcol,nonleaf,length}
       \begin{itemeq}
          \frac{|Z_{L'l}^{W'u}(G)|}{|Z_{L'l'}^{W'u}(G)|}   \leq(1+\alpha/\Delta)^d(1+c/\Delta)^{\Delta-d}\,.
       \end{itemeq}
 \end{enumerate} 
\end{enumerate}
\end{restatable}

Note that Statement~\ref{itm:kcol,nonzero} with $W=L=\emptyset$ is the result we want for Theorem~\ref{thm:zero-free}, Statement~\ref{itm:kcol,leaf} shows that changing the fixed color of a leaf (degree 1) vertex $u$ affects the angle and length of the restricted partition function by a small amount (depending on whether the neighbour of $u$ is itself free or fixed), and Statement~\ref{itm:kcol,nonleaf} is a similar but weaker version for any vertex. 

\subsection{Proof of Theorem~\ref{thm:zero-free}}\label{app:zero-free,details}

Lemma~\ref{lem:kcol} directly implies Theorem~\ref{thm:zero-free}, and we give the proof in this section.
%
First, we state a lemma of Barvinok which is useful for evaluating sums of restricted partition functions.
\begin{lemma}[{Barvinok~\cite[Lemma 3.6.3]{Barvinok16}}]
 \label{lem:cone}
 Let $u_1, \ldots, u_n \in \RR^2$ be non-zero vectors such that the angle between any two vectors $u_i$ and $u_j$ is at most $\alpha$ for some $\alpha \in \interval[co]{0}{2\pi/3}$.
 Then the $u_i$ all lie in a cone of angle at most $\alpha$ and
 \[
 \bigg| \sum_{i=1}^n u_i \bigg| \geq \cos(\alpha/2) \sum_{i=1}^n |u_i|.
 \]
\end{lemma}
\noindent
Furthermore the following simple corollary of of the cosine rule will come in handy.
\begin{lemma}
\label{lem:smallanglediff}
 Let $z, z'$ be two complex numbers at an angle of at most $\pi/3$, then $|z-z'| \leq \max \{ |z|, |z'| \}$.
\end{lemma}
\begin{proof}
 Recall the cosine rule, for a triangle with sides $a$, $b$ and $c$; and angles $A$, $B$ and $C$ where side $a$ is not adjacent to angle $A$, then
 \[
 |a|^2 = |b|^2 + |c|^2 -2|b||c| \cos(A)\,,
 \]
 where $|a|$ is the length of side $a$.
 Now consider the triangle with vertices in $\CC$ at the origin, $z$ and $z'$.
 The sides have length $|z|, |z'|$ and $|z-z'|$ and the angle at the origin is the angle $\theta \leq \pi/3$ between $z$ and $z'$.
 As $\cos(x) \geq 1/2$ for $x \leq \pi/3$,
 \[
 |z-z'|^2 \leq |z|^2+|z'|^2-|z||z'| \leq \max\{|z|^2,|z'|^2\}\,.\qedhere
 \]
\end{proof}
 
To prove Lemma~\ref{lem:kcol} we need some definitions and an auxiliary lemma. 
We define rational functions (which depend on $k$ and $w$) in two variables $z_0,z$ and respectively $k-1$ variables $z_0,\ldots,z_{k-2}$ by
\begin{align*}
 R(z_0,z;w,k) &:= \frac{wz_0+(k-2)z+1}{z_0+(k-2)z+w}\,,\\
 R_k(z_0,z_1,\ldots,z_{k-2};w) &:= \frac{wz_0+z_1+\ldots+z_{k-2}+1}{z_0+z_1+\ldots+z_{z-2}+w}\,.
\end{align*}
Consider the cone
\[
C(\theta):=\{z=re^{i\vartheta}\mid r\geq 0 \text{ and } |\vartheta|\leq \theta\},
\]
and define for $d=0,\dotsc,\Delta$ and $c,\alpha> 0$, the region
\[
K(\theta,d,c,\alpha,\eps):=C(d\theta+\Delta-\eps)\cap \left\{z: \left(1+\frac{c}{\Delta}\right)^{d-\Delta}\left(1+\frac{\alpha}{\Delta}\right)^{d}\leq |z|\le\left(1+\frac{c}{\Delta}\right)^{\Delta-d}\left(1+\frac{\alpha}{\Delta}\right)^{d}\right\}\,.
\]

\begin{lemma}
\label{lem:forward invariant}
Let $\Delta\in \mathbb{N}_{\geq 3}$ and let $k\in \mathbb{N}_{\geq 3}$. Define $c=\log k-1$ and $\alpha=\log k^{1/2}-1$. Then there exists $0<\eps<\theta<\pi/(3\Delta)$ and $\eta=\omega(\frac{1}{\Delta})$
such that for each $d=0,\ldots,\Delta$, and any $z_0,\ldots,z_{k-2}\in K_d:=K(\theta,d,c,\alpha,\eps)$ such that for each $i,j$, $z_i/z_j\in K_d$ and any  $w\in \mathcal{N}([1,1+c/\Delta],\eta)$ the ratio $R=R_k(z_0,z_1,\ldots,z_{k-2};w)$ satisfies
\begin{align}
(1+\alpha/\Delta)^{-1}<|R|&< 1+\alpha/\Delta\quad \text{ and }\quad  |\arg(R)|< \theta.\label{eq:ratio bound arg}
\end{align}
In particular the values $\theta = 1/(5 \Delta)$, $\eps = \theta / [100 \log k]$ and $\eta = \min\{\Delta c / [800(\Delta + \alpha)^2], 1/[2400(\Delta+ \alpha)], c / [800 \Delta]\}$ suffice.
\end{lemma}
\noindent
We will prove this lemma in the next subsection, but we first utilize it to prove Lemma~\ref{lem:kcol}, which we restate here for convenience.

\zerofreeinduction*

\begin{proof}
The choice of constants is the same as in Lemma~\ref{lem:forward invariant} except that we need to choose $\eta$ small enough so that each $w\in \mathcal{N}([1,1+c/\Delta],\eta)$ has argument at most $\eps$. It thus suffices to take $\eta=\omega(\frac{1}{\Delta\log(k)})$.

We prove the lemma by induction on the number of free vertices of $G$.
For the base case, we have no free vertices and so every vertex is fixed.
Therefore $Z_L^W(G)$ is a product of non-zero terms, hence is non-zero, proving~\ref{itm:kcol,nonzero}.
Statement \ref{itm:kcol,leaf}\ref{itm:kcol,leaf,free} is vacuous as there are no free vertices.
Statement \ref{itm:kcol,leaf}\ref{itm:kcol,leaf,fixed} follows as the products $Z_{L'l}^{W'u}(G)$ and $Z_{L'l'}^{W'u}(G)$ differ in at most one term.
Thus their ratio is either $1$, $w$ or $w^{-1}$.
Similarly we deduce Statement~\ref{itm:kcol,nonleaf} (in which $d$ must be zero) from the fact that the products $Z_{L'l}^{W'u}(G)$ and $Z_{L'l'}^{W'u}(G)$ differ in at most $\Delta$ terms.

Now, we assume that Statements~\ref{itm:kcol,nonzero}, \ref{itm:kcol,leaf}, and~\ref{itm:kcol,nonleaf} hold for graphs with $r \ge 0$ free vertices. We prove the statements for $r+1$ free vertices. First, we shall prove~\ref{itm:kcol,nonzero}.

Suppose that $u$ is a free vertex.
Note that $Z_L^W(G) = \sum_{j=1}^k Z_{Lj}^{Wu}(G)$.
As each term in the sum on the right hand side of this expression has one fewer free vertex, we may apply induction to deduce that all of these terms are non-zero by \ref{itm:kcol,nonzero}.
Furthermore, by \ref{itm:kcol,nonleaf} each pair has angle at most $d\theta +(\Delta-d)\eps$ where $d$ is the number of free neighbours of $u$.
Lemma~\ref{lem:cone} tells us that the $Z_{Lj}^{Wu}$ all lie in a cone of angle at most $d\theta+(\Delta-d)\eps$ and
\begin{equation*}
 |Z_L^W(G)| = \bigg|\sum_{j=1}^k Z_{Lj}^{Wu}(G) \bigg| \geq \cos(d\theta/2+(\Delta-d)\eps/2)\sum_{j=1}^k |Z_{Lj}^{Wu}(G)| \neq 0\,.
\end{equation*}

Next, we shall prove~\ref{itm:kcol,leaf}\ref{itm:kcol,leaf,free} so consider the ratios,
\begin{align*}
 R_{j,l}(G) = \frac{Z_{L'j}^{W'u}(G)}{Z_{L'\ell}^{W'u}(G)}, & & R_{j,\ell}^v(G) = \frac{Z_{L'j}^{W'v}(G- u)}{Z_{L'\ell}^{W'v}(G - u)}.
\end{align*}
As $v$ is the unique neighbour of $u$ and is free, we may write, denoting $j^*$ for $\pi_{uv}(j)$ and $\ell^*$ for $\pi_{uv}(\ell)$,
\[
 R_{j,l}(G) = \frac{\sum_{i}Z^{Wuv}_{Lji}(G)}{\sum_{i}Z^{Wuv}_{L\ell i}(G)}
 =\frac{wZ^{Wv}_{Lj^*}(G-u)\sum_{i\notin\{ j^*,\ell^*\}}Z^{Wv}_{Li}(G-u)+Z^{Wv}_{L\ell^*}(G-u)}{Z^{Wv}_{Lj^*}(G-u)+\sum_{i\notin\{ j^*,\ell^*\}}Z^{Wv}_{Li}(G-u)+wZ^{Wv}_{L\ell^*}(G-u)}.
 \]
 Dividing both the numerator and denominator by $Z^{Wv}_{L\ell^*}(G-u)$ (which by induction is nonzero) we obtain,
 \begin{equation}
 \dfrac{w R_{j^*,\ell^*}^v(G) + \sum_{i \neq j^*, \ell^*} R_{i,\ell^*}^v(G) +1}{R_{j^*,\ell^*}^v(G) +\sum_{i \neq j^*,\ell^*} R_{i,l}^v(G)+w} = R_k(R_{j^*,\ell^*}^v(G), R_{1,\ell^*}^v(G), \dotsc, R_{k,\ell^*}^v(G);w)\,.
 \label{eq:recurrence}
\end{equation}
Where the function $R_k$ in~(\ref{eq:recurrence}) takes as arguments all $R^v_{i,\ell^*}(G)$ for $i \neq \ell^*$ precisely once (and so takes precisely $k-1$ arguments as expected.)

Suppose that $v$ has $d$ free neighbours that are not $u$.
Since $G -u$ has one fewer free vertex than $G$, we may apply the inductive hypothesis.
By \ref{itm:kcol,nonleaf} we find that for any $i\neq \ell^*$, we have $R^v_{i,\ell}(G)\in K_d$.
However, we also have that for any $i,j\neq \ell^*$, that 
\[
\frac{R^v_{i,\ell^*}(G)}{R^v_{j,\ell^*}(G)}=\frac{Z_{L'i}^{W'v}(G- u)}{Z_{L'j}^{W'v}(G- u)}=R^v_{i,j}(G)\in K_d.
\]
To prove~\ref{itm:kcol,leaf}\ref{itm:kcol,leaf,free}\ref{itm:kcol,leaf,free,angle}, observe that the angle between $Z_{L'j}^{W'u}$ and $Z_{L'l}^{W'u}$ is precisely the angle of $R_{j,l}(G)$ from the real axis in $\CC$ and so is bounded by the absolute value of the argument of $R_{j,l}(G)$, which by Lemma~\ref{lem:forward invariant} bounded by $\theta$ as desired.
Statement~\ref{itm:kcol,leaf}\ref{itm:kcol,leaf,free}\ref{itm:kcol,leaf,free,length} also follows immediately from Lemma~\ref{lem:forward invariant}.

For the proof of~\ref{itm:kcol,leaf}\ref{itm:kcol,leaf,fixed}, we note that as $v$ is fixed, then 
$$
Z_{L'j}^{W'u}(G) \in \{ w^{-1} Z_{L'l}^{W'u}(G),Z_{L'l}^{W'u}(G),wZ_{L'l}^{W'u}(G) \}
$$
from which both~\ref{itm:kcol,leaf,fixed,angle} and~\ref{itm:kcol,leaf,fixed,length} follow.

Finally, we prove \ref{itm:kcol,nonleaf}.
To do so we consider the graph $G \star u$ which is formed as follows.
Let $v_1, \ldots, v_r$ be the neighbours of $u$ ordered arbitrarily.
Let $u_1, \ldots, u_r$ be $r$ new vertices which will be copies of $u$.
Then $G \star u$ is the graph obtained by deleting $u$ and its incident edges, adding the vertices $u_1, \ldots, u_r$ and edges $u_1v_1, \ldots, u_rv_r$.
Furthermore, $G \star u$ inherits any colouring of $G$ and if $u$ is coloured, all of the new vertices inherit this colour.
Note that if $u$ is coloured, then the graph $G \star u$ has the same partition function as $G$.
Also, in this case $G \star u$ has the same number of free vertices as $G$.
This allows us to prove \ref{itm:kcol,nonleaf} from \ref{itm:kcol,leaf} by changing the colour of one copy of $u$ at a time.
That is,
\begin{equation}
 \frac{Z_{L'j}^{W'u}(G)}{Z_{L'l}^{W'u}(G)} = \frac{Z_{L'j\ldots j}^{W'u_1\ldots u_r}(G\star u )}{Z_{L'l\ldots l}^{W'u_1\ldots u_r}(G\star u )} = \prod_{i=1}^r \frac{Z_{L'j\ldots j l \ldots l}^{W'u_1\ldots u_{i-1}u_i \ldots u_r}(G\star u )}{Z_{L'j\ldots j l \ldots l}^{W'u_1\ldots u_i u_{i+1} \ldots u_r}(G\star u )} \label{eq:telescoping}
\end{equation}
By \ref{itm:kcol,leaf} each of the terms in the product in~(\ref{eq:telescoping}) has angle at most $\theta$ and absolute value at most $1 +\alpha/\Delta$ (if $u_i$ is free) or angle at most $\eps$ and absolute value at most $1 +c/\Delta$ (if $u_i$ is fixed).
As $u$ has $d$ free neighbours and at most $\Delta -d$ fixed neighbours, this allows us to conclude~\ref{itm:kcol,nonleaf}\ref{itm:kcol,nonleaf,angle} and \ref{itm:kcol,nonleaf}\ref{itm:kcol,nonleaf,length}, completing the induction.
\end{proof}

\subsection{Proof of Lemma~\ref{lem:forward invariant}}
\label{sec:boundcalpha}

We will require a technical lemma which concerns the real and imaginary parts of the ratios $R(z_1,z_2;w,k)$.
\begin{lemma}
\label{lem:reimparts}
 Let $z_1, z_2 \in \CC$ be defined as $z_1 = xe^{i\thx}$, $z_2 = ye^{i\thy}$ with $x,y \in \RR^+$ and $\thx,\thy \in \interval[co]{0}{2\pi}$ and suppose $w \in[1, 1 + \frac{c}{\Delta}]$ is real.
 Then, the real and imaginary parts of $R(z_1,z_2;w,k)$ are as follows where $N$ is a nonzero constant,
 \begin{align}
 \label{eq:realR} \Re(R(z_1,z_2;w,k)) &= N(wx^2+(w+1)(k-2)xy\cos(\thx-\thy) +(k-2)^2y^2 \\ &\qquad+(w^2+1)x\cos(\thx)+(w+1)(k-2)y\cos(\thy) +w), \nonumber\\
 \label{eq:imagR} \Im(R(z_1,z_2;w,k)) &= N (w-1)((k-2)xy\sin(\thx-\thy) +(1+w)x\sin(\thx) +(k-2)y\sin(\thy)).
 \end{align}
\end{lemma}
\begin{remark} 
\label{rem:reimparts}
Set $\theta =\max(|\thx|,|\thy|, |\thx-\thy|)$ and assume $|\theta| \le 1$. Then as $|\sin t| \leq |t|$ and $\cos t \geq 1 - t^2/2$ for all $t$, and using $w \geq 1$ we obtain the following bounds:
 \begin{align*}
 \Re(R(z_1,z_2;w,k))  
 &\ge N(1 - \theta^2/2)(wx^2 + (w+1)(k-2)xy +(k-2)^2y^2  \\
&\qquad +(w^2+1)x +(w+1)(k-2)y +w) \\
 &\ge N(1 - \theta^2/2)(x+(k-2)y+w)(wx+(k-2)y+1);\\
 \text{and}\\
  \Im(R(z_1,z_2;w,k)) & \le N (w-1)((k-2)xy|\thx-\thy| +(1+w)x|\thx| +(k-2)y|\thy|).
 \end{align*} 
Hence
\begin{align}
 \bigg|\frac{\Im(R(z_1,z_2;w,k))}{\Re(R(z_1,z_2;w,k))}\bigg| \leq
 \frac{(w-1)\left((k-2)xy|\thx-\thy|+(1+w)x|\thx|+(k-2)y|\thy|\right)}
 {(1 - \frac{\theta^2}{2})(x+(k-2)y+w)(wx+(k-2)y+1)}. 
\label{eq:imreratio}
 \end{align}
\end{remark}

\begin{proof}
 We may write $z_1 = x\cos(\thx) +ix\sin(\thx)$ and $z_2 = y\cos(\thy) +iy\sin(\thy)$. Hence,
 \begin{align}
  \nonumber R(z_1,z_2;w,k) & = \frac{w(x\cos(\thx) +ix\sin(\thx)) +(k-2)(y\cos(\thy) +iy\sin(\thy)) +1}{x\cos(\thx) +ix\sin(\thx) +(k-2)(y\cos(\thy) +iy\sin(\thy)) +w} \\
  & = \frac{wx\cos(\thx) +(k-2)y\cos(\thy) +1 + i(wx\sin(\thx) + (k-2)y\sin(\thy))}{x\cos(\thx) +(k-2)y\cos(\thy) +w + i(x\sin(\thx) + (k-2)y\sin(\thy))} \label{eq:R(z1z2wk)1}
 \end{align}
Rationalising the denominator in~(\ref{eq:R(z1z2wk)1}), we obtain
\begin{align}
 \nonumber R(z_1,z_2;w,k) = & N^{-1}\left(wx\cos(\thx) +(k-2)y\cos(\thy) +1 + i(wx\sin(\thx) + (k-2)y\sin(\thy))\right)\\
 & \times \left(x\cos(\thx) +(k-2)y\cos(\thy) +w - i(x\sin(\thx) + (k-2)y\sin(\thy))\right) \label{eq:R(z1z2wk)2}
\end{align}
where $N = |x\cos(\thx) +(k-2)y\cos(\thy) +w + i(x\sin(\thx) + (k-2)y\sin(\thy))|^2$.
Expanding the expression in~(\ref{eq:R(z1z2wk)2}), the real and imaginary parts are given by the following expressions in which we write $c_x$ for $\cos(\thx)$ and similarly define $c_y,s_x$ and $s_y$ to simplify notation.
 \begin{align*}
  \Re(R(z_1,z_2;w,k)) &= N^{-1}(wx^2c_x^2+(w+1)(k-2)xyc_xc_y+(k-2)^2c_y^2 \\
  &\qquad + wx^2s_x^2 + (w+1)(k-2)xys_xs_y+(k-2)^2s_y^2\\
  &\qquad +(w^2+1)xc_x+(w+1)(k-2)yc_y+w)\\
  \Im(R(z_1,z_2;w,k)) &= N^{-1}((k-2)xy(c_xs_y+ws_xc_y) - (k-2)xy(wc_xs_y+s_xc_y)\\
  &\qquad + (w^2-1)xs_x+(w-1)(k-2)ys_y )
 \end{align*}
Combining these expressions with the trigonometric identities
\begin{gather*}
 \cos^2(\vartheta)+\sin^2(\vartheta) = 1 \\
 \sin(\alpha-\beta) = \sin(\alpha)\cos(\beta)-\sin(\beta)\cos(\beta) \\
 \cos(\alpha-\beta) = \cos(\alpha)\cos(\beta)+\sin(\alpha)\sin(\beta)
\end{gather*}
yields the expressions~(\ref{eq:realR}) and ~(\ref{eq:imagR}) as claimed.

By an application of the triangle law combined with an applications of the approximations, $|\sin\theta| \leq |\theta|$ and $\cos\theta \geq 1-\theta^2/2$, we obtain
\begin{align}
 |\Im(R(z_1,z_2;w,k))| &\leq N^{-1} (w-1)\left((k-2)xy|\thx-\thy|+(1+w)x|\thx|+(k-2)y|\thy|\right), \label{eq:Imbound} \\
 \Re(R(z_1,z_2;w,k)) &\geq N^{-1} \left((wx+(k-2)y+1)(x+(k-2)y+w) \right.\nonumber\\ 
&\qquad -\left. ((w+1)(k-2)(x+1)y+(w^2+1)x)\theta^2/2\right). \label{eq:Rebound1}
\end{align}
Dividing \eqref{eq:Imbound} by \eqref{eq:Rebound1}, noting that for $\theta$ small this is maximised when $w = 1+\frac{c}{\Delta}$ and regrouping some terms yields the bound \eqref{eq:imreratio}.
\end{proof}
\noindent
We can now give a proof of Lemma~\ref{lem:forward invariant}

\begin{proof}[Proof of Lemma~\ref{lem:forward invariant}]
We first prove a slightly stronger version of the lemma for $w'$ real.
That is, we will show that there exist a small constants $\kappa = c/100$ and $\kappa'=0.02$ such that
\begin{equation}\label{eq:real case}
(1+(\alpha-\kappa)/\Delta)^{-1}<|R|<1+(\alpha-\kappa)/\Delta \quad \text{ and}\quad |\arg(R)|< (1-\kappa')\theta.
\end{equation}
To do so we start by taking a constant $\delta$ small enough so that for all $k \geq 3$, $c$ and $\alpha$ satisfy the strict inequality 
\begin{equation}
 \frac{ce^c}{cos(\delta)(e^c+k-1)} < \alpha - \kappa; 
 \label{eq:calphabounds}
\end{equation}
for example $\delta = 1/2$ is sufficient.

Fix $d\in\{0,\ldots,\Delta\}$.
First we observe that we may assume that $|R|\geq 1$. 
Indeed, if $|R|<1$, then 
\[
1/R=\frac{z_0+\sum_{i=1}^{k-2}z_i+w}{wz_0+\sum_{i=1}^{k-2}z_i+1}=\frac{1+\sum_{i=1}^{k-2}z_i/z_0+w/z_0}{w+\sum_{i=1}^{k-2}z_i/z_0+1/z_0}
\]
and $|1/R|>1$.
Since for each $i,j\geq 0$, the pairs $z_i/z_0$ and $z_j/z_0$ also satisfy our assumptions this shows our claim.
We start by showing that $|R|$ is bounded by $1+(\alpha - \kappa)/\Delta$.

We observe that (setting $z = (z_1 + \cdots + z_{k-2})/k$)
\begin{equation}\label{eq:from 1}
|R| = |R(z_0,z;w,k)|=\bigg|1+\frac{(w-1)z_0+(1-w)}{z_0+\sum_{i=1}^{k-2}z_i+w}\bigg|\leq 1+\frac{\frac{c}{\Delta}|z_0-1|}{|z_0+\sum_{i=1}^{k-2}z_i+w|}.
\end{equation}
Lower bounding the denominator of~(\ref{eq:from 1}) may be done with an application of Barvinok's lemma.
For the numerator we apply Lemma~\ref{lem:smallanglediff} as the angle between $z_0$ and $1$ is certainly less than $\pi/3$. 
This allows us to deduce that
\[
 |R(z_0,z;w,k)| \leq 1 +\frac{\frac{c}{\Delta} \max\{|z_0|,1\}}{\cos(d\theta/2+(\Delta-d)\eps/2))(|z_0|+\sum_{i=1}^{k-2}|z_i|+1)}.
\]
We next observe that by symmetry we may assume that $|z_0|\leq 1$; otherwise we divide the numerator and the denominator by $z_0$.
To maximize the above quantity clearly one should take each $|z_i|$ as small as possible. 
So we take $|z_i| = (1+c/\Delta)^{d-\Delta}(1+\alpha/\Delta)^{-d} \geq (1+c/\Delta)^{-\Delta} \geq e^{-c}$ and noting that $\theta\leq \delta/\Delta$, we rearrange to deduce that 
\[
|R(z_0,z;w,k)| < 1+\frac{c / \Delta}{\cos(\delta)((k-1)e^{-c}+1)}< 1+(\alpha-\kappa)/\Delta
\]
by~(\ref{eq:calphabounds}). This proves the first bound in \eqref{eq:real case}.


For the other bound in \eqref{eq:real case}, recall that
 $z=\frac{1}{k-2}\sum_{i=1}^{k-2}z_j$ so that
 $R_k(z_0,z_1,\ldots,z_{k-2};w)=R(z_0,z;w,k)$. 
Note that $z\in C(d\theta+(\Delta-d)\eps)$ by convexity of the cone and so by Lemma~\ref{lem:cone} we have
\[
\cos(d\theta/2+(\Delta-d)\eps/2)(1+c/\Delta)^{d-\Delta}(1+\alpha/\Delta)^{-d}\leq |z|\leq (1+c/\Delta)^{\Delta-d}(1+\alpha/\Delta)^{d}.
\]
To prove the bound on the argument of $R(z_0,z;w,k)$ we use the inequality, $|\beta|\leq |\tan(\beta)|$. It therefore suffices to bound the ratio $\frac{|\Im R(z_0,z;w,k)|}{|\Re R(z_0,z;w,k)|} = \tan(\arg(R(z_0,z;w,k)))$, which by Lemma~\ref{lem:reimparts} and Remark~\ref{rem:reimparts} is bounded by
 \begin{equation}\label{eq:bound arg}
  \frac{(w-1)\left((k-2)|z_0z||\theta_0-\theta_z|+(1+w)|z_0 \theta_0|+(k-2)|z\theta_z|\right)}
 {(1 - \frac{\theta^2}{2})(|z_0|+(k-2)|z|+w)(w|z_0|+(k-2)|z|+1)}.  
\end{equation}
Now suppose we can prove that
\begin{equation}\label{eq:strict angle}
  \frac{\left((k-2)|z_0z||\theta_0-\theta_z|+(1+w)|z_0 \theta_0|+(k-2)|z\theta_z|\right)}
 {(|z_0|+(k-2)|z|+w)(w|z_0|+(k-2)|z|+1)}
  < \frac{\Delta \tau \theta}{c} 
\end{equation}
where $\tau  = 7/e^2 < 0.96$.
Then by choosing $\theta \leq 0.2$ and $\kappa' < 0.02$, we have that $\Delta \tau \theta /c < (1 - \frac{\theta^2}{2})(w-1)^{-1}(1 - \kappa')\theta$ (using $w < 1 + c/\Delta$). This together with \eqref{eq:strict angle}  proves that \eqref{eq:bound arg} is at most $\theta(1-\kappa')$ and hence $|\arg(R(z_0,z;w,k))|< \theta(1-\kappa')$, as desired.

We will now show that \eqref{eq:strict angle} holds.
So, first note that
\begin{align}
 \nonumber & \frac{\left((k-2)|z_0z||\theta_0-\theta_z|+(1+w)|z_0 \theta_0|+(k-2)|z\theta_z|\right)}
 {(|z_0|+(k-2)|z|+w)(w|z_0|+(k-2)|z|+1)} \\
 & \leq \frac{\left((k-2)|z_0z||\theta_0-\theta_z|+2|z_0\theta_0|+(k-2)|z\theta_z|\right)}{(|z_0|+(k-2)|z|+1)^2}, \label{eq:anglemax}
\end{align}
which can be observed by computing the derivative of the left hand side of~(\ref{eq:strict angle}) with respect to $w$ and noting it is strictly negative for $w \geq 1$.
Now, we maximize~(\ref{eq:anglemax}), so first we show that there is a maximum point where exactly two of $|\theta_0-\theta_z|,|\theta_0|,|\theta_z|$ are as large as possible and one is zero.
To see this, first note that clearly at least one of $|\theta_0 - \theta_z|, |\theta_0|, |\theta_z|$ must be as large as possible i.e. equal to $d\theta +(\Delta-d)\eps$.
In fact exactly two of these must be maximised as the maximization with respect to the $\theta$ terms only is of the form $f(\theta_0,\theta_z)=a|\theta_0 - \theta_z|+b|\theta_0| +c|\theta_z|$ for constants $a,b,c>0$.
So if $|\theta_0 - \theta_z| = d\theta +(\Delta-d)\eps$ for example, then if $b\geq c$ we may set $\theta_0=d\theta +(\Delta-d)\eps$, $\theta_z=0$ increasing $f(\theta_0,\theta_z)$.
Similar logic allows one to conclude that two of $|\theta_0 - \theta_z|, |\theta_0|, |\theta_z|$ are equal to $d\theta +(\Delta-d)\eps$ and one is $0$ in every other case.

This leaves us with three maximization problems over $R_d \subseteq \RR^2$ defined by 
\begin{align*}
 R_d = \{(x,y) |  (1+c/\Delta)^{d-\Delta}(1+\alpha/\Delta)^{-d} \leq & x \leq(1+c/\Delta)^{\Delta-d}(1+\alpha/\Delta)^{d} , \\
  \cos(d\theta/2+(\Delta-d)\eps/2)(1+c/\Delta)^{d-\Delta}(1+\alpha/\Delta)^{-d} \leq & y \leq (1+c/\Delta)^{\Delta-d}(1+\alpha/\Delta)^{d}, \\
  \cos(d\theta/2+(\Delta-d)\eps/2)(1+c/\Delta)^{d-\Delta}(1+\alpha/\Delta)^{-d} \leq & y/x\leq (1+c/\Delta)^{\Delta-d}(1+\alpha/\Delta)^{d}\}.
\end{align*}
We enlarge the region slightly obtaining the region $\widetilde{R}_d \subseteq \RR^2$ defined by
\begin{align*}
 \widetilde{R}_d & = \{(x,y) | \cos(\delta)exp(-(\frac{d}{\Delta} \alpha + (1-\frac{d}{\Delta})c)) \leq x, y, y/x \leq  exp(\frac{d}{\Delta} \alpha + (1-\frac{d}{\Delta})c)\cos(\delta)^{-1}\} \\
 & = \{(x,y) | \cos(\delta)e k^{\frac{d}{2\Delta}-1} \leq x, y, y/x \leq  k^{1-\frac{d}{2\Delta}}(e\cos(\delta))^{-1} \}
\end{align*}
The functions to maximise are,
\begin{align*}
 f_1(x,y) & =\frac{(k-2)(xy+y)}{(x+(k-2)y+1)^2}, \\
 f_2(x,y) & =\frac{(k-2)xy+2x}{(x+(k-2)y+1)^2}, \\
 f_3(x,y) & =\frac{2x+(k-2)y}{(x+(k-2)y+1)^2}.
\end{align*}
First we look at $f_1$, it has critical points along the line $x+1=(k-2)y$ where it attains its maximum value of $1/4$.
However, note that due to our choice of $c$ and $\alpha$, this line does not lie inside of $\widetilde{R_d}$, hence the maximum must be attained at a boundary point.
Furthermore both $f_2$ and $f_3$ have no critical points strictly inside the first quadrant, so again their maxima must be attained at a boundary point.
This allows us to reduce the problem to eighteen univariate maximization problems, each of which has maximum at most $3 e^{-1}k^{-\frac{d}{2\Delta}}$ over $\widetilde{R}_d$ (see Section~\ref{app:zfr} for details).

Thus~(\ref{eq:anglemax})  is upper bounded by  $(d\theta+(\Delta-d)\eps)3 e^{-1}k^{-\frac{d}{2\Delta}}$. As a function of $d$, this is maximised when $d = (\frac{2}{\log k} - \frac{\eps}{\theta - \eps})\Delta$, which  (if $d \geq 1$) gives an upper bound to \eqref{eq:anglemax} of 
\[
\frac{6e^{-2}\Delta(\theta - \eps)}{\log k} \exp\left(  \frac{\eps}{2(\theta - \eps)} \log k\right). 
\]
Thus \eqref{eq:strict angle} is satisfied provided $\frac{6}{7}(\theta - \eps)\exp( \frac{1}{2}\log k \frac{\eps}{\theta - \eps}) < \theta$. By taking $\eps = \theta x / \log k$ and assuming $\log k \geq 1$, the left hand side is bounded above by $\frac{6}{7} (1-x)\exp(x/2(1-x))\theta$ and this is at most $\theta$ (as required) by taking $x = 1/100$ as assumed in the statement of the lemma.   
If $d=0$, then as $f_1, f_2$ and $f_3$ are all bounded above by $1$, provided $\eps < \frac{\theta}{100\log(k)}$, the left hand side of~(\ref{eq:strict angle}) at most $\eps \Delta < \tau \theta \Delta/c$.
This completes the proof of~(\ref{eq:strict angle}) and hence of \eqref{eq:real case}.

We finally extend the proof to the case that $w\in \mathcal{N}([1,c/\Delta],\eta)$ 
for $\eta= 1/[800(\Delta + \alpha)^2]$ using continuity.
First observe that $R_k(w):=R_k(z_0, \ldots, z_{k-2};w)$ satisfies 
\[
R_k(w) = z_0 + \frac{(z_0 + (k-2)z + 1)(1 - z_0)}{z_0 + (k-2)z + w}.
\] 
Then 
\[
|R_k(w + \eta) - R_k(w)| 
= \left| \frac{[z_0 + (k-2)z + 1](1 - z_0)}{(z_0 + (k-2)z + w + \eta)(z_0 + (k-2)z + w)}\eta \right| 
\]
The numerator is upper bounded by  $[|z_0| + (k-2)|z| + 1](1 + |z_0|)|\eta|$, while the denominator is lower bounded by 
\[
\big[(|z_0| + (k-2)|z| + |w| - |\eta|\cos^{-1}(\Delta \theta/2))(|z_0| + (k-2)|z| + |w|)\cos(\Delta \theta/2)\big]^2
\]
where we use the fact that the angle between any two of $w, z_0, z$ is at most $\Delta \theta$ and so we can apply Barvinok's lemma. In the statement of the lemma, we assume $\Delta \theta \leq \pi / 3$ so $\cos(\Delta \theta) \geq 1/2$. Then using that $(x+a)/(x+b) \leq a/b$ for $x\geq 0$ and $a\geq b$ and using that $|w| \geq 1$ and that $\eta < 1/4$ (so that $|w| - |\eta|\cos^{-1}(\Delta \theta/2) > 1/2$), we have
\[
\frac{|z_0| + (k-2)|z| + 1}{(|z_0| + (k-2)|z| + |w| - |\eta| \cos^{-1}(\Delta \theta) ) \cos(\Delta \theta)}
\leq 4
 \:\:\:\text{and}\:\:\:
\frac{|z_0| +  1}{(|z_0| + (k-2)|z| + |w|)\cos(\Delta \theta)}
\leq 2.
\]
Combining the above inequalities we obtain $|R_k(w + \eta) - R_k(w)| \leq 8 \eta$. Recall $\eta \leq \min\{\Delta c / [800(\Delta + \alpha)^2], 1/[2400(\Delta+ \alpha)], c / [800 \Delta]\}$.
Then for $w\in \mathcal{N}([1,c/\Delta],\eta)$, we can write $w = w' + \eta$ with $w' \in [1,  \frac{c}{\Delta}]$ real. Writing $R = R(w)$, we have 
\[\left( 1 + \frac{\alpha}{\Delta} \right)^{-1}
\leq \left( 1 + \frac{\alpha - \kappa}{\Delta} \right)^{-1} - 8\eta
\leq |R(w')| - 8\eta
<|R| < |R(w)| + 8\eta 
\leq \left( 1 + \frac{\alpha - \kappa}{\Delta} \right) + 8\eta
\leq 
1 + \frac{\alpha}{\Delta},
\]
where the first and last inequalities follow by our choice of $\eta$.

It follows from simple geometry that if $|z_1 - z_2| \leq \mu$ for $z_1, z_2 \in \mathbb{C}$ and $\mu \in \mathbb{R}^+$ with $\mu < |z_1|$, then $|\arg(z_1) - \arg(z_2)| < \arcsin(\mu / |z_1|)$. Using this, and since $|R(w')|> (1 + \frac{\alpha}{\Delta})^{-1}$, we see that 
\[
\arg(R) < \arg(R(w')) + \arcsin\left( 8 \eta \left(1 + \frac{\alpha}{\Delta} \right) \right) = \theta(1 - \kappa')  + \arcsin\left( 8 \eta \left(1 + \frac{\alpha}{\Delta} \right) \right) < \theta;
\]
in order to check the last inequality holds, it is sufficient 
to check that $8 \eta (1 + \frac{\alpha}{\Delta}) < \sin(\kappa' \theta)$. 
Noting that $\sin x > x - x^3/6 > 5x/6$ for $x \in (0,1)$, 
it is sufficient that $ 8 \eta (1 + \frac{\alpha}{\Delta}) < 5\kappa'\theta / 6 $ (since $\kappa' \theta <1$ by our choice of $\kappa'$ and $\theta$) and this holds by our choice of $\eta$. This completes the proof of the lemma.
\end{proof}

\subsection{Maximization problems}
\label{app:zfr}
We look at the maximization problems coming from Section~\ref{sec:boundcalpha} and claim that each has an upper bound of at most $3 k^{-\frac{d}{2\Delta}}/e$.
We find eighteen of them, one for each of the three functions with either $x$, $y$, or $y/x$ fixed to one of the two corresponding boundary values.
This allows us to reduce to the univariate maximization problems detailed below.
To simplify the expressions we will let $r=k-2 $, $s = \cos(\delta)e k^{\frac{d}{2\Delta}-1}$ and $t=k^{1-\frac{d}{2\Delta}}(e\cos(\delta))^{-1}$.

\begin{center}
\begin{tabular}{ |c|c|c|c| }
\hline
& $f_1$ & $f_2$ & $f_3$ \\ \hline
  $x = s$ 
  & $p_1(y) =\frac{ry(1+s)}{(s+ry+1)^2}$
  & $p_2(y) = \frac{rsy+2s}{(s+ry+1)^2}$
  & $p_3(y) = \frac{2s+ry}{(s+ry+1)^2}$
  \\ \hline
  
  $x = t$ 
  & $p_4(y) = \frac{ry(1+t)}{(t+ry+1)^2}$ 
  & $p_5(y) = \frac{rty+2t}{(t+ry+1)^2}$
  & $p_6(y) = \frac{2t+ry}{(t+ry+1)^2}$
  \\ \hline
  
  $y = s$ 
  & $p_7(x) = \frac{rs(1+x)}{(x+rs+1)^2}$ 
  & $p_8(x) = \frac{rxs+2x}{(x+rs+1)^2}$
  & $p_9(x) = \frac{2x+rs}{(x+rs+1)^2}$
  \\ \hline
  
  $y = t$ 
  & $p_{10}(x) = \frac{rt(1+x)}{(x+rt+1)^2}$ 
  & $p_{11}(x) = \frac{rxt+2x}{(x+rt+1)^2}$
  & $p_{12}(x) = \frac{2x+rt}{(x+rt+1)^2}$
  \\ \hline
  
  $y/x = s$ 
  & $p_{13}(x) = \frac{rs(1+x^{-1})}{(x^{-1}+rs+1)^2}$ 
  & $p_{14}(x) = \frac{rs+2x^{-1}}{(x^{-1}+rs+1)^2}$ 
  & $p_{15}(x) = \frac{2x^{-1}+rx^{-1}s}{(x^{-1}+rs+1)^2}$ 
  \\ \hline
  
  $y/x = t$ 
  & $p_{16}(x) = \frac{rt(1+x^{-1})}{(x^{-1}+rt+1)^2}$ 
  & $p_{17}(x) = \frac{rt+2x^{-1}}{(x^{-1}+rt+1)^2}$ 
  & $p_{18}(x) = \frac{2x^{-1}+rx^{-1}t}{(x^{-1}+rt+1)^2}$  
  \\ \hline
\end{tabular}
\end{center}

To begin the maximization, first observe that under the map $x \mapsto x^{-1}$, each of the functions $p_j(x)$ is the same as some function $p_l(x)$ for some $13 \leq j \leq 18$ and $7 \leq l \leq 12$.
Furthermore, $y=s$ yields the bounds $s \leq x \leq 1$ and $y/x = s$ gives $1 \leq x \leq t$.
Similarly we may compare $y=t$ and $y/x = t$.
Thus the ranges for $x$ are identical after inverting $x$.
Hence we may ignore $p_{13}$ through $p_{18}$ leaving us with $12$ problems.

Next, consider $p_{10}$, $p_{11}$ and $p_{12}$, each of which can be bounded above by
$$
\frac{2rtx}{(x+rt+1)^2} \leq \frac{2rtx}{r^2t^2} \leq \frac{2}{r}
$$
where the final inequality follows as $x \leq t$.

Similarly, we can bound $p_4$, $p_5$ and $p_6$.
As it must be the case that $y \geq 1$, the numerator of each is bounded above by $2try$.
Thus an upper bound for all three is $2t/ry$.
Furthermore, $r \geq \frac{2k}{3\cos(\delta)}$ provided $k\geq 7$ and $\delta$ small enough. So we are left with an upper bound of $3 k^{-\frac{d}{2\Delta}}/e$.

The remaining problems are similar.
The numerators may all be bounded above by $rs(1+x) \leq 2rs$ (or for $p_1, p_2$ and $p_3$ by $2ry$.)
The denominators are all bounded from below by $r^2s^2$ and $r^2y^2$ respectively.
Thus all six of these are upper bounded by $2/rs$ which is at most $3 k^{-\frac{d}{2\Delta}}/e$.

Hence an upper bound on all of the problems $p_1$ through $p_{18}$ is $3 k^{-\frac{d}{2\Delta}}/e$ as claimed.

\subsection{Improvements for small \texorpdfstring{$k$}{k}}
When $k$ is small, then the parameter $c=\log(k)-1$ is also very small.
In fact we do not obtain a better constant than what is known for the Ising model until $k \geq 21$.
However it is possible to do better, we can choose different values for $\alpha$ and $c$ which work better in these cases.
In this section we will show how to derive such values.

First, we note that we may do the the analysis in an identical way until we find ourselves with the maximization problems $f_1, f_2$ and $f_3$.
Now we maximise these more carefully than in Section~\ref{app:zfr}.
First, for $f_1$ we apply AM-GM to the denominator to deduce that $f_1(x,y) \leq \frac{1}{4}$ for any $x,y$.
This allows us to take any $c<4$ and as $k$ is small this is all we need and so we may ignore this constraint.
This leaves us to maximise $f_2$ and $f_3$.
A similar argument to the one in the proof of Lemma~\ref{lem:forward invariant} allows us to deduce that the maxima are on the boundary of $R_d$ and hence we need only consider the boundary of $\widetilde{R_d}$.

Now, we proceed as in Section~\ref{app:zfr} with different choices of $s$ and $t$ where this time we will take $t=e^{d/\Delta \alpha + (1-d/\Delta)c}$ and $s=t^{-1}$.
We start with $12$ maximization problems which we reduce to $8$ by symmetry as before.
Furthermore, $f_2>f_3$ if and only if $x>1$ which allows us to halve the number of problems left to consider leaving us with $4$ problems.
More precisely, we are left with $p_3, p_5, p_9$ and $p_{11}$.
All of these are of the form $f(x)=(a_1 x+a_2)(x+a_3)^{-2}$ which has a maximum at $x=a_3-2a_2/a_1$.
See the following table for the maximization of these $4$ functions.

\begin{center}
\begin{tabular}{ |c|c|c|c|c|c| }
\hline
Function & $a_1$ & $a_2$ & $a_3$ & $x^*$ & $f(x^*)$ 
\\ \hline
$p_3$ & $\frac{1}{k-2}$ & $\frac{2s}{(k-2)^2}$ & $\frac{s+1}{k-2}$ & $\frac{1-3s}{k-2}$ & $\frac{1}{4(1-s)}$
\\ \hline
$p_5$ & $\frac{t}{k-2}$ & $\frac{2t}{(k-2)^2}$ & $\frac{t+1}{k-2}$ & $\frac{t-3}{k-2}\leq 1$ & $\frac{kt}{(t+k-1)^2}$
\\ \hline
$p_9$ & $2$ & $(k-2)s$ & $1+(k-2)s$ & $1$ & $\frac{1}{(2+(k-2)s)}$
\\ \hline
$p_{11}$ & $2+(k-2)t$ & $0$ & $1+(k-2)t$ & $1+(k-2)t>t$ & $\frac{(k-2)t^2+2t}{((k-1)t+1)^2}$
\\ \hline
\end{tabular}
\end{center}
Note that in the cases of $p_5$ and $p_{11}$ the maximum value $x_*$ is outside the domain which we are maximising over and thus we maximise at the endpoints of the domain instead.

Now, recall that the maximum values obtained above must also satisfy~(\ref{eq:calphabounds}).
Also, when $s=e^{-\alpha}$ it must be the case that $(2+(k-2)e^{-\alpha})^{-1} <c^{-1}$ (from $p_9$).
Combining these after rearrangement yields the inequity
\begin{equation}
 \frac{ce^c}{e^c+k-1} \leq \alpha \leq \log\bigg(\frac{k-2}{c-2}\bigg)
 \label{eq:howtofindca}
\end{equation}
We may solve this inequality computationally for $c$, and deduce that there is a choice of $\alpha, c$ which satisfies~(\ref{eq:howtofindca}) provided that $c \leq c_k$ for some $c_k$ which can be found in the following table.
The corresponding value of $\alpha_k$ is also provided.
We give both $c_k$ and $\alpha_k$ rounded to three decimal places.

\begin{center}
 \begin{tabular}{c||c|c|c|c|c|c|c|c|c|c}  $k$ & $3$ & $4$ & $5$ & $6$ & $7$ & $8$ & $9$ & $10$ & $11$ & $12$ \\
 \hline
  $\alpha_k$ & $1.767$ & $1.803$ & $1.849$ & $1.896$ & $1.944$ & $1.990$ & $2.034$ & $2.076$ & $2.116$ & $2.154$ \\  \hline
  $c_k$ & $2.171$ & $2.330$ & $2.472$ & $2.600$ & $2.716$ & $2.820$ & $2.916$ & $3.003$ & $3.084$ & $3.160$ \\
 \end{tabular}
\end{center}

Now, we check that these are indeed the maximum values. To do this, we first note that we have $p_3 \leq 1/4$ and applying AM-GM to the denominator of the maximum for $p_5$ yields a result which is smaller than the values from we obtained for the maximum of $p_9$.
Finally, for $p_{11}$, the denominator is at least $(k-1)(k-2) t^2 +2t(k-1)$. Thus, after cancellations we are left with $p_{11} \leq 1/(k-1)$ which suffices for $k \geq 4$.
For $k=3$ we can easily check that $(t^2+2t)(2t+1)^{-2}$ is maximised when $t=1$ and hence is certainly at most $1/3 < 1/2.17$.

Recall when computing the maximum of $p_9$, we took $s$ as large as possible where one would expect that we should do the opposite to maximise $p_9$.
We now justify this choice.
So recall that we must ensure $d\theta p_9(x) \leq \Delta \theta/c$.
Furthermore, $s$ may be considered as a function of $d$ and as such is equal to $\exp(-d/\Delta \alpha - (1-d/\Delta)c)$.
Thus we must ensure that
$$
g(d)=\frac{dc/\Delta}{2+(k-2)s} \leq 1.
$$
Writing $\lambda$ for $d/\Delta$ gives the following function with domain $[0,1]$
$$
G(\lambda)=\frac{\lambda c}{2+(k-2)e^{-\lambda \alpha -(1-\lambda)c}}.
$$
Differentiating this with respect to $\lambda$, we see that either $c-\alpha<1$ and $G$ is increasing on $[0,1]$ or there is a maximum with $\lambda >1$ which is not inside the domain.
Thus, we maximise $G$ at one of its boundary points and it is easy to see that $\lambda=1$ is the maximum point rather than $\lambda =0$ where $G(\lambda)=0$.

\section{Cluster expansion}\label{sec:cluster}

On the surface our proof of Theorem~\ref{thm:alg,cluster} has a similar flavour to the polynomial interpolation method: we define a series expansion for $\log Z(G;w)$, show that it converges, and approximate $Z(G;w)$ by computing the coefficients of a truncation of the series. 
Instead of working with a Taylor series and a zero-free region, we work with a different formal power series for $\log Z$ called the \emph{cluster expansion} which expresses $\log Z$ as a sum of some weights over connected subgraphs of $G$. 
This technique was recently applied to approximating the partition functions of the Potts and random cluster models in~\cite{BCHPT19}, where the \emph{random cluster model} is random graph model from statistical physics that generalizes the Ising and Potts models, and percolation\footnote{Note the distinct uses of the term `cluster' in `cluster expansion' and `random cluster model', though there is a common theme of connected subgraphs in both uses.}.
To obtain the result we adapt a standard reduction to express our partition function $Z(G;w)$ in terms of the random cluster model, and apply the method of~\cite{BCHPT19} which gives an approximation algorithm via the cluster expansion.

\subsection{The random cluster model}

The random cluster model, instead of counting graph labelings according to satisfied edges, counts connected subgraphs according to some weights. 
We adapt the standard reduction comparing the Potts model and random cluster model partition functions to our $Z(G;w)$ for Unique Games instances.

We start by rewriting $Z(G;w)$ for a given instance $G=(V,E,\pi)$. 
Let $(V',E')$ be a connected component of $G$, so that $V'$ is a nonempty subset of $V$ and $E'\subset E\cap \binom{V'}{2}$. 
We consider a singleton vertex $\{u\}$ to comprise the connected component $(\{u\},\emptyset)$. 
Define
\[
\sat_\pi(V',E'):=\sum_{\{x_u\}_{u\in V'}\in[k]^{V'}}\;\prod_{(u,v)\in E'}\indicator{x_v=\pi_{uv}(x_u)},
\]
where $\indicator{P}=1$ if $P$ is true, and $0$ otherwise. 
In other words $\sat_\pi(V',F')$ counts the number of assignments of value $1$ (perfectly satisfying assignments) of the Unique Games instance restricted to the subgraph $(V',E')$. 
The definition means that $\sat_\pi(\{u\},\emptyset)=k$ as there are no constraints and the empty product is $1$.
Since we work with $(V',E')$ connected, we also have 
\begin{equation}\label{eq:con}
0\le \sat_\pi(V', E')\le k\,,
\end{equation}
as given any starting color for an arbitrary first vertex $u\in V'$, there is at most one completion of the coloring to a perfectly satisfying assignment obtained by following the constraints out along the component from $u$. 

We use the notation $\comp(V, F)$ for the set of connected components of the graph $(V, F)$, taken as pairs $(V', E')$ with $E'\subset F$. 
The following lemma gives the reduction from $Z(G;w)$ to the random cluster model partition function.

\begin{restatable}{lemma}{randomcluster}\label{lem:random cluster}
Let $G=(V,E,\pi)$ be a UG instance and $w\in\CC$.
Then 
\[
Z(G;w)=\sum_{F\subseteq E} (w-1)^{|F|} \prod_{(V',E')\in\comp(V,F)} \sat_{\pi}(V',E').
\]
\end{restatable}

\begin{proof}
This follows by writing $w=1+(w-1)$ and expanding the partition function:
\begin{align*}
Z(G; w) 
  &= \sum_{\{x_u\}_{u\in V}\in[k]^V} \prod_{\substack{(u,v)\in E,\\ x_v=\pi_{uv}(x_u)}}(1 + (w-1))
\\&= \sum_{\{x_u\}_{u\in V}\in[k]^V} \prod_{(u,v)\in E}\Big(1 + (w-1)\indicator{x_v=\pi_{uv}(x_u)}\Big)
\\&= \sum_{\{x_u\}_{u\in V}\in[k]^V} \sum_{F\subset E}\prod_{(u,v)\in F}(w-1)\indicator{x_v=\pi_{uv}(x_u)}
\\&= \sum_{F\subset E}(w-1)^{|F|}\sum_{\{x_u\}_{u\in V}\in[k]^V} \prod_{(u,v)\in F}\indicator{x_v=\pi_{uv}(x_u)}\,,
\end{align*}
where the second line follows from writing the product over all edges instead of just satisfied edges, the third line follows by expanding the product, writing $F$ for the edges for which the term $(w-1)\indicator{x_v=\pi_{uv}(x_u)}$ is taken, and the final line follows by interchanging the order of summation.
Now if we break the final sum over color assignments and product over satisfied edges into a sum and product for each component $(V',E')$ of $(V,F)$, and recall the definition of $\sat_\pi(V',E')$, we obtain
\[
Z(G;w) = \sum_{F\subset E}(w-1)^{|F|}\prod_{(V',E')\in\comp(V,F)}\sat_\pi(V',E')\,.\qedhere
\]
\end{proof}

\subsection{The cluster expansion}

We closely follow the notation and setup of~\cite{BCHPT19,HPR19}. Given a UG instance $G=(V,E,\pi)$, define a \emph{polymer} $\gamma$ to be a connected subgraph of $G$ with at least two vertices. A collection of polymers is \emph{compatible} if the polymers contained in it are pairwise vertex disjoint.
We define the incompatibility graph $H_G$ on the collection of all polymers as follows: vertices of $H_G$ are the polymers and two polymers are connected by an edge if they are not compatible (that is if they share a vertex). 
Write $|\gamma|:=|V(\gamma)|$, $\|\gamma \|:=|E(\gamma)|$, and given $w\in\CC$, define the \emph{weight} of a polymer $\gamma$ as 
\[
w_\gamma:=(w-1)^{\|\gamma\|}k^{-|\gamma|}\sat_\pi(\gamma)\,,
\]
where we write $\sat_\pi(\gamma)$ for the more cumbersome $\sat_\pi(V(\gamma),E(\gamma))$.
Then by Lemma~\ref{lem:random cluster} and the observation that for a single vertex $u$ we have $\sat_\pi(\{u\},\emptyset)=k$, we have 
\[
\Xi(G):=\sum_{\Gamma=\{\gamma_1,\ldots,\gamma_t\}}\prod_{i=1}^t w_{\gamma_i}=k^{-|V|}Z(G;w),
\]
where the sum is over all sets $\Gamma$ of (pairwise) compatible polymers. 
Note that $\Xi(G)$ is the multivariate independence polynomial of the compatibility graph $H_G$.

The \emph{cluster expansion} is the following formal power series for $\log\Xi(G)$:
\begin{equation}\label{eq:cluster expansion}
\log\Xi(G)=\sum_{\substack{\Gamma\subset V(H_G)\\\mathclap{H_G[\Gamma] \text{ connected}}}}\phi(\Gamma)\;\prod_{\gamma\in \Gamma}w_\gamma,
\end{equation}
where $\phi(\Gamma)$ is the \emph{Ursell function} of the graph $H_G[\Gamma]=(\Gamma,F)$, defined as
\[
\phi(\Gamma):=
\frac{1}{|\Gamma|!}\sum_{\substack{A\subseteq F\\ (\Gamma, A)\text{ connected}}} (-1)^{|A|}.
\]
For $\Gamma\subset V(H_G)$, let $\|\Gamma\|$ be given by $\|\Gamma\|:=\sum_{\gamma\in \Gamma}\|\gamma\|$, and define the truncated cluster expansion as follows
\begin{equation}
T_m:=\sum_{\substack{\mathclap{\Gamma\subset V(H_G),\, \|\Gamma\|< m}\\H_G[\Gamma] \text{connected}}}\phi(\Gamma)\;\prod_{\gamma\in \Gamma}w_\gamma\,.
\end{equation}

With the definitions and a reduction to the right partition function in place, we can now state the conditons of~\cite{BCHPT19} that imply the cluster expansion converges and gives an approximation guarantee.

\begin{restatable}[Borgs et al.~\cite{BCHPT19}]{lemma}{KPconvergence}\label{lem:KPconvergence}
Suppose that polymers are connected subgraphs of a graph $G$ of maximum degree $\Delta$ on $n$ vertices. Suppose further that for some $b>0$ and all polymers $\gamma$ the following hold:
\begin{align}
\|\gamma\| &\geq b|\gamma|\,,\text{ and} \label{eq:ce 1}
\\
|w_\gamma| &\leq e^{-\big(\frac{3\log\Delta}{b}+3\big)\|\gamma\|}\,. \label{eq:ce 2}
\end{align}
Then the cluster expansion converges absolutely and for any $m\in \NN$, $|T_m-\log\Xi(G)|\leq n e^{-3m}$.
\end{restatable}

\noindent
To prove Theorem~\ref{thm:alg,cluster} we simply check that these conditions hold, which we state as a lemma below. 

\begin{restatable}{lemma}{clusterconvergence}\label{lem:clusterconvergence}
Let $\Delta\in \mathbb{N}_{\geq 16}$, let $C=e^{-9-2\log\Delta}$, and let $\zeta= 8\sqrt{1/\Delta}$.
Then if $k\geq C^{-2\Delta/\zeta}$ and $1\leq w\leq e^{(2-\zeta)\log(k)/\Delta}$, Lemma~\ref{lem:KPconvergence} holds for \UG[k] instances $G$ of maximum degree $\Delta$.
\end{restatable}

\begin{proof}[Proof of Lemma~\ref{lem:clusterconvergence}]
We proceed in a manner inspired by~\cite[Theorem 2.4]{BCHPT19}.
First observe that for \eqref{eq:ce 1} we can take $b=1/2$ since polymers are connected and have size at least $2$. We next check \eqref{eq:ce 2} with $b=1/2$, showing that the conditions $k\geq C^{-2\Delta/\zeta}$ and $1\leq w\leq e^{(2-\zeta)\log(k)/\Delta}$ give
\begin{equation}\label{eq:weight bound}
|w_\gamma|\leq C^{\|\gamma\|}\,.
\end{equation}

We verify \eqref{eq:weight bound} in three cases according to the value of $s=\|\gamma\|$. We also recall the bound~\eqref{eq:con}.
\begin{description}
\item[Case 1:] $s>2\Delta/\zeta$. We use that $|\gamma|\geq 2\|\gamma\|/\Delta$.
Then
\begin{align*}
|w_\gamma|&=k^{-|\gamma|}(w-1)^{\|\gamma\|}\sat_\pi(\gamma)\leq k^{-|\gamma|}k(w-1)^{\|\gamma\|} \leq k^{-\|\gamma\|/(2\Delta)}k(w-1)^{\|\gamma\|}
\\
&\leq k^{-s2/\Delta} k k^{s(2-\zeta)/\Delta} \leq k^{1-s\zeta/\Delta}\leq k^{-s\zeta/(2\Delta)},
\end{align*}
which is bounded by $C^{-s}$ since $k\geq C^{-2\Delta/\zeta}$.
\item[Case 2:] $\Delta<s\leq 2\Delta/\zeta$. We use that fact that $\|\gamma\|\leq \binom{|\gamma|}{2}$ and thus $\sqrt{2s}< |\gamma|$. 
Then
\[
|w_\gamma|\leq k k^{(2-\zeta)s/\Delta}k^{-\sqrt{2s}}=k^{1+(2-\zeta)s/\Delta-\sqrt{2s}}.
\]
Looking at the exponent of $k$ we see by our assumptions on $s$ that
\[
1+(2-\zeta)s/\Delta-\sqrt{2s}\leq 1+4/\zeta-\sqrt{2\Delta}\leq 1-\sqrt{\Delta}/2\leq -1
\]
for $\Delta$ large enough (i.e $\Delta\geq 16$ suffices).
So since $k\geq C^{-2\Delta/\zeta}$ we are in business.
\item[Case 3:] $1\leq s\leq \Delta$. If $|\gamma|= 2$ we have $s=1$ and therefore
\[
|w_\gamma|\leq k^{-1}k^{(2-\zeta)/\Delta}\leq k^{-1/2}
\]
provided $\Delta\geq 4$.
If $|\gamma|\geq 3$ we have
\[
|w_\gamma|\leq k^{-2}k^{(2-\zeta)}= k^{-\zeta}.
\]
So since $k\geq C^{-2\Delta/\zeta}$ the required bound holds.
\end{description}
This finishes the proof.
\end{proof}

We deduce the following runtime guarantees from our setup and the analyses of~\cite{HPR19,PR17}.
The truncated series $T_m$ can be computed in time $e^{O(\Delta m+\log n)}$ given an enumeration of all polymers on fewer than $m$ edges and their weights (see~\cite{HPR19}). 
We can enumerate the polymers in time $O(n^2m^7(e\Delta)^{2m})$ as they are connected subgraphs of a graph of maximum degree $\Delta$ (see~\cite{PR17}), and compute each weight in time $O(km)$ as all perfectly satisfying assignments on a connected graph are found by following each of the $k$ assignments of an initial vertex and propagating along constraints.
To get an approximation of the form $e^{-\alpha}\le Z(G;w^*)/\xi\le e^\alpha$ we take $m=\log(n/\alpha)/3$ which means the entire computation of $\xi$ can be done in time
\[
e^{O(\Delta m+\log n)} + O(km^8n^2(e\Delta)^{2m})= kn^{O(1)}(n/\alpha)^{O(\Delta)}\,.
\]
We see that the number of colours needed to make the lemma work is of the order $\Delta^{O(\Delta^{3/2})}$. It would be interesting to get a better dependence on $\Delta$.

\section{Conclusions}\label{sec:conc}

Lemma~\ref{lem:solveCUG} shows that a hypothetical polynomial-time algorithm for computing $Z(G;w)$ exactly when 
\[
\log w = \frac{2}{1-\eps-\delta}\frac{\log k}{\Delta}
\]
would refute the UGC. 
This problem is likely \#P hard so we resort to approximation, which we can only do for some range of $w$.
In Theorem~\ref{thm:zero-free} we have $\log w = (1-o(1))\log(k)/\Delta$ when $k$ is small enough that $\log k=o(\Delta)$, and in Theorem~\ref{thm:alg,cluster} we have $\log w = (2-o(1))\log(k)/\Delta$ when $k$ is larger than some $\Delta^{\poly(\Delta)}$. 
This means we must increase $f$ from $k^{-1}$ to solve any CUG problems, and the size of $w$ in these results is what gives the bound on $f$ in Theorems~\ref{thm:main,interp} and~\ref{thm:main,cluster}.
It is therefore important to determine the threshold $f^*$ such that \CUG[f] is equivalent to UG when $f\le f^*$. 
Trivially we have $k^{-1}\le f^*\le 1$, but if one could show e.g.\ that $f^*\ge k^{2\theta-1}$ then Theorem~\ref{thm:main,cluster} would mean the Unique Games problem is in P for bounded-degree graphs when $k$ is large enough and $\eps+\delta < \theta$.

\subsection{Phase transitions}

In this subsection we focus on the ferromagnetic Potts model, which is the special case of our partition function $Z(G;w)$ when the constraints on each edge are the identity permutation and we take $w\ge 1$.
The behaviour of the Potts model on bounded-degree graphs is strongly related to the \emph{phases} of the model on the infinite $\Delta$-regular tree $\TT[\Delta]$. 
We will not define precisely what we mean by a phase or a phase transition here, but as the parameter $w$ varies, the behaviour of the model undergoes certain changes that seem to affect both zeros of the partition function and the dynamics of associated Markov chains.
Häggström~\cite{Haggstrom96} showed that the \emph{uniqueness phase transition} on $\TT[\Delta]$ occurs at $w=w_u(k,\Delta)$, the unique value of $w$ for which the polynomial 
\begin{equation}\label{eq:def,wu}
(k-1)x^\Delta + (2-w-k)x^{\Delta-1}+wx-1
\end{equation}
has a double root in $(0,1)$. 
There is a further \emph{ordered/disordered phase transition} (see~\cite{GSVY16}) at 
\[
w_o(k,\Delta):=\frac{k-2}{(k-1)^{1-2/\Delta}-1}\,.
\]
Below we relate the values of $w$ found in Theorems~\ref{thm:alg,cluster} and~\ref{thm:zero-free} to these phase transitions.

\subsection{Potential barriers to improving Theorem~\ref{thm:main,interp}}\label{sec:barrers,interp}

To strengthen Theorem~\ref{thm:main,interp} to a result that would refute the UGC, we need roughly a factor two improvement in the leading constant in $\log w^* \sim \log (k)/\Delta$ as $k\to\infty$ (provided $\log k=o(\Delta)$) from Theorem~\ref{thm:zero-free}, but there are reasons to believe it may be hard to make such an improvement. 

The authors of~\cite{BGP16} analysed a natural Markov chain known as the \emph{Glauber dynamics} which walks the set of possible colorings of a graph $G$, and when this mixes rapidly we expect an efficient, randomised approximation algorithm for the Potts model partition function to follow.
They showed for the Potts model that Glauber dynamics mixes rapidly on graphs of maximum degree $\Delta$ when 
\[
\log w \le (1+o(1))\frac{\log k}{\Delta-1}\,,
\]
as $k\to\infty$, and that on almost all $\Delta$-regular graphs (for $\Delta\ge 3$), Glauber dynamics mixes slowly when $w$ is just a little larger, satisfying
\[
\log w > (1+o(1))\frac{\log k}{\Delta-1-\frac{1}{\Delta-1}}\,.
\]
These bounds sandwich the phase transition point $w_u$; they also showed that as $k\to\infty$, 
\[
\log w_u = \frac{\log k}{\Delta-1} + O(1)\,.
\]
Thus it appears that $w_u$ is a barrier for approximating $Z(G;w)$ via Glauber dynamics. 

Similarly, we expect that $w_u$ is a barrier for the zero-free region. 
Underpinning Lemma~\ref{lem:kcol} is a complex dynamical system (see~\cite{PR18} for a treatment of the case $k=2$), and equation~\eqref{eq:def,wu} appears in the analysis of this system.
Essentially, the behaviour of a fixed point in the complex dynamics changes at $w=w_u$ in a way which means it is reasonable to expect zeros of $Z(G;w)$ to accumulate near $w_u$ for some $G$ with maximum degree $\Delta$. 
Thus we suspect that the method cannot work for $w>w_u$.

\subsection{Potential barriers to improving Theorem~\ref{thm:main,cluster}}\label{sec:barrers,cluster}

There are several regimes of interest for the algorithm in Theorem~\ref{thm:alg,cluster} that gives Theorem~\ref{thm:main,cluster}.
When we apply Theorem~\ref{thm:alg,cluster} to solve CUG problems, we only need an approximation with $\alpha = Cn$ for constant $C$ in which case the running time is bounded by $kn^{O(1)}e^{O(\Delta)}$. 
Recall that we also need $k\ge \Delta^{O(\Delta^{3/2}})$, and hence when $\Delta$ and $k$ do not grow too fast with $n$, the running time is sub-exponential in $n$. 
For the hypercube with $\Delta=\log n$ and with $k$ as small as the result allows, the algorithm is quasi-polynomial.
In the case where $\Delta$ is constant, to get an approximation as accurate as $\alpha$ being constant the running time of the algorithm is polynomial in $n$.
It is interesting to compare this with a paper of Galanis et al.~\cite{GSVY16} who show that it is \#BIS-hard to approximate the partition function of the Potts model with $k$ colours on graphs of maximum degree $\Delta$ when $w>w_o$. 
With $\zeta=8/\sqrt\Delta$, if we take $k=k_0=\exp((18\Delta+4\Delta\log\Delta)/\zeta)$, then Theorem~\ref{thm:alg,cluster} shows that we can approximate the Potts model partition function on graphs of maximum degree at most $\Delta$ with $w=k_0^{(2-\zeta)/\Delta}$. 
A quick calculation shows that, as $\Delta\to \infty$ (and hence $k\to\infty$), $w_o(k_0,\Delta)\sim k_0^{2/\Delta}$. 
We conclude that if one assumes that there are no efficient algorithms for approximating \#BIS-hard problems, Theorem~\ref{thm:alg,cluster} is very close to optimal in this regime.

\section{Acknowledgements}

Part of this work was done while E.\ Davies, A.\ Kolla, and G.\ Regts were visiting the Simons Institute for the Theory of Computing.
Part of this work was done while M.\ Coulson was visiting V.\ Patel and G.\ Regts at the University of Amsterdam. This visit was funded by a Universitas 21 Travel Grant from the University of Birmingham.
M.\ Coulson was supported by the Spanish Ministerio de Economía y Competitividad project MTM2017--82166--P and an FPI Predoctoral Fellowship from the Spanish Ministerio de Economía y Competitividad with reference PRE2018--083621.
V. Patel is partially supported by the Netherlands Organisation for Scientific Research (NWO) through Gravitation-grant NETWORKS-024.002.003.
We would like to thank Tyler Helmuth for insightful discussions.

\appendix
\section{Details for the proof of Theorem~\ref{thm:alg,interp}}\label{app:interp,details}

What follows is a rather precise description of results developed by Barvinok in~\cite{Barvinok16,Barvinok18}, lightly specialised to our application and notation.

\begin{definition}\label{def:taylor}
Let $f:\CC\to\CC$ be any function, and $m\ge 0$. Then we define the \emph{degree-$m$ Taylor polynomial of $f$ about zero}, $T_m(f)$, as the polynomial in $z$ given by
\[
T_m(f)(z) := f(0) + \sum_{k=1}^m\frac{f^{(k)}(0)}{k!}z^k\,.
\]
\end{definition}

\begin{lemma}[{see~\cite[Lemma 2.2.1]{Barvinok16} or \cite[Lemma 2.1]{Barvinok18}}]\label{lem:tayloraccuracy}
Let $g:\hCC\to\hCC$ be a polynomial of degree at most $N$, and for $\beta>1$ suppose that $g(z)\ne 0$ for $|z|<\beta$.

Then given a choice of branch for $f(z)=\log g(z)$ where $|z|<\beta$, we have 
\[
|f(1)-T_m(f)(1)| \le \frac{N}{(m+1)\beta^m(\beta-1)}\,.
\]
\end{lemma}

\begin{corollary}[{cf.~\cite[Corollary 2.2]{Barvinok18}}]\label{cor:tayloreps}
For any $c>0$ there exists $c'>0$ such that the following holds.
Suppose that the conditions of Lemma~\ref{lem:tayloraccuracy} hold, and in addition that $\beta\le 1+c$. 

Then for any $0<\alpha<N/e$, and for any 
\[
m \ge \frac{c'}{\beta-1}\log\left(\frac{N}{\alpha}\right)\,,
\]
we have $|f(1)-T_m(f)(1)|\le \alpha$.
\end{corollary}

The only differences from~\cite[Corollary 2.2]{Barvinok18} are the relaxation of the assumption $\alpha<1$ to $\alpha\le N/e$, the additional assumption that $\beta$ is close to $1$, and a more precise analysis of $m$. 
In fact one can take $c'=c/\log(1+c)$.

\begin{proof}
By Lemma~\ref{lem:tayloraccuracy}, we are done if 
\[
\frac{N}{(m+1)\beta^m(\beta-1)} \le \alpha
\]
for $m$ as in the statement of the corollary. This holds if and only if
\[
(m+1)\beta^{m+1} \ge \frac{N}{\alpha}\frac{\beta}{\beta-1} \Longleftrightarrow
(m+1)\log\beta \ge W\left(\frac{N}{\alpha}\frac{\beta\log\beta}{\beta-1}\right)\,,
\]
where $W$ is the upper real branch of the Lambert $W$-function, see~\cite{CGHJK96}. We take this branch because $\beta> 1$ so $(m+1)\log\beta > 0$.
Since $W$ is increasing and $\frac{\log \beta}{\beta-1} <1$, this is implied by 
\[
m+1 \ge \frac{W(N\beta/\alpha)}{\log\beta}\,.
\]
Now we have $\log x \ge W(x)$ for all $x\ge e$, so since $\beta>1$ and $N/\alpha \ge e$ this is implied by 
\[
m+1 \ge \frac{\log(N\beta/\alpha)}{\log\beta} = \frac{\log(N/\alpha)}{\log\beta} + 1\,.
\]
Then it suffices to take $m\ge \log(N/\alpha)/\log\beta$. But since $\log\beta\sim\beta-1$ as $\beta\to1$ and $\beta< 1+c$, we can find $c'$ depending only on $c$ such that $c'\log\beta \ge \beta -1$. Then it suffices to take 
\[
m \ge \frac{c'}{\beta-1}\log\left(\frac{N}{\alpha}\right)\,\qedhere
\]
\end{proof}

Corollary~\ref{cor:tayloreps} tells us how many terms of a Taylor expansion of $\log g$ we need to get an additive error of at most $\alpha$, under the condition that $g$ has no roots in the disc $\{z\in\CC : |z|<\beta\}$. 
We want to work with a zero-free region of the form $\cN([0,1],\eta)$, the open set containing a ball of radius $\eta$ around each point in $[0,1]$. 
Barvinok~\cite{Barvinok16,Barvinok18} also gives constructions that perform well for this situation which we reproduce below.

\begin{lemma}[{\cite[Lemma 2.2.3]{Barvinok16}}]\label{lem:disctostrip}
For $0<\rho<1$ there is a polynomial $p$ of degree
\[
N_\rho := \left\lfloor\left(1+\frac{1}{\rho}\right)e^{1+\frac{1}{\rho}}\right\rfloor \ge 14
\]
such that $p(0)=0$, $p(1)=1$, and 
\[
-\rho \le \Re\, p(z) \le 1 + 2\rho 
\quad\text{and}\quad
|\Im\, p(z)| \le 2\rho
\]
for all $z$ such that $|z|\le \beta(\rho)$ where
\[
\beta_\rho := \frac{1-e^{-1-\frac{1}{\rho}}}{1-e^{-\frac{1}{\rho}}} > 1\,.
\]
\end{lemma}

\begin{corollary}[{cf.~\cite[Theorem 1.6]{Barvinok18}}]\label{cor:stripaccuracy}
Suppose that $0<\eta<1$, and $g$ is a polynomial of degree $N$ such that $g(z)\ne0$ for all $z\in\cN([0,1],\eta)$. 
Then given any $0<\alpha<Ne^{6/\eta-1}$, it suffices to compute the first
\[
m=e^{6/\eta}\log\left(\frac{Ne^{6/\eta}}{\alpha}\right)
\]
coefficients of $g$ to obtain a number $\xi$ satisfying $|\log g(1)-\xi| \le \eps$. 
\end{corollary}
\begin{proof}
Let $\rho=\eta/\sqrt 8$ and $\beta_\rho,N_\rho$ be given by Lemma~\ref{lem:disctostrip}. 
Then since $\eta<1$, we note that
\[
N_\rho\le e^{6/\eta}
\quad\text{and}\quad
\beta_\rho \ge 1 + \frac{1}{2e^{\frac{1}{\rho}}} \ge 1+ e^{-4/\eta}\,.
\]
Now the polynomial $p$ as in Lemma~\ref{lem:disctostrip} maps $\{z\in\CC : |z|\le \beta_\rho\}$ into $\cN([0,1], \eta)$, so the polynomial $g\circ p(z)$ is a degree $NN_\rho \le Ne^{6/\eta}$
polynomial which is nonzero for all $z\in\CC$ such that $|z|\le 1+ e^{-4/\eta}$.

We now apply Corollary~\ref{cor:tayloreps} to $g\circ p$. 
With $f(z) = \log(g\circ p(z))$ we have $f(1) = \log g(1)$ since $p(1)=1$, and so the Taylor polynomial $T_m(1)$ (as defined in Definition~\ref{def:taylor}) for this $f$ is the quantity we want for $\xi$.
More precisely, as described in~\cite[Section 2.2.2]{Barvinok16}, to compute $T_m(\log g\circ p)$ at $z=0$ it suffices to compute $T_m(g\circ p)$. 
In turn, to compute $T_m(g\circ p)$ it suffices to compute $T_m(g)$ and the truncation $p_m$ of $p$ obtained by deleting all monomials of degree higher than $m$, and then the composition of polynomials $T_m(g) \circ p_m$.
The final step is to truncate $T_m(g) \circ p_m$ to be degree $m$, and to obtain $\xi$ by evaluating this polynomial at $z=1$.
This can be done in time $O(m)$, and by Corollary~\ref{cor:tayloreps} applied to $g\circ p$, when we have
\[
m \ge e^{6/\eta}\log\left(\frac{Ne^{6/\eta}}{\alpha}\right)\,,
\]
and $\alpha < Ne^{6/\eta-1}$, we have the desired $|\log g(1) - \xi| \le\alpha$.
\end{proof}

We know now how many coefficients are needed for an approximation to a polynomial. 
For the complexity of computing these coefficients we refer to Patel and Regts~\cite{PR17,PR19}.
Our partition function $Z(G;w)$ is an edge-coloured BIGCP in the sense of Patel and Regts' definition in~\cite{PR19}.
$Z(G;w)$ has degree at most $\Delta n/2$ with BIGCP parameters $\alpha=2$ and $\beta_i=O(k^i)$ according to~\cite[Section 6]{PR17}.
Then by~\cite[Theorem 2.1]{PR19} there is a deterministic algorithm to compute the first $m$ coefficients of $Z$ in time
\[
\tilde O\big(n(4e\Delta\sqrt k)^{2m}\big)\,,
\]
where $\tilde O$ means that we omit factors polynomial in $m$.

To prove Theorem~\ref{thm:alg,interp} we want an approximation for $\log Z(G;w)$ with error $\alpha=C n$, and we have a zero-free region surrounding $[1,w^*]$ with $w^*=1+(\log k-1)\Delta$ at distance $\eta=\omega(1/(\Delta\log k))$. 
Then we can transform $Z$ into a polynomial with zero-free region around $[0,1]$ of distance $\eta/(w^*-1) = \omega(1/(\log k)^2)$, so we need
\[
C  \le \frac{2}{e\Delta}e^{O(\log^2 k)} \,,
\]
and
\[
m = e^{O(\log^2 k)}\log\left(\frac{\Delta}{2C}e^{O(\log^2 k)}\right)\,,
\]
according to Corollary~\ref{cor:stripaccuracy}.

Then we have a running time of
\[
\tilde O\big(n(4e\Delta\sqrt k)^{2m}\big)
= n\exp\left(e^{O(\log^2 k)} \log\left(\frac{\Delta}{C}e^{O(\log^2 k)}\right)\log(\Delta\sqrt k)\right)\,.
\]

\bibliographystyle{habbrv}
\bibliography{FerroPotts-UGC}

\end{document}